\documentclass[11pt]{article}

\usepackage{typearea}
\typearea{14}
\usepackage{setspace}

\usepackage{bbm}
\usepackage{framed} 
\usepackage{url} 
\usepackage{complexity}
\usepackage{booktabs}
\usepackage{amsmath,amssymb}
\usepackage{amsthm}
\usepackage{thmtools} 
\usepackage{thm-restate}
\usepackage{nicefrac}
\usepackage{microtype}
\usepackage{calc}
\usepackage{enumerate}
\usepackage{enumitem}
\usepackage[usenames,dvipsnames]{xcolor}
\usepackage[colorlinks,citecolor=blue,linkcolor=BrickRed]{hyperref}
\usepackage{algorithm}
\usepackage{graphicx}
\usepackage[font=footnotesize,labelfont=bf]{subcaption}
\usepackage[font=footnotesize,labelfont=bf]{caption}
\usepackage[nobreak=true]{mdframed}
\usepackage{appendix}
\usepackage[noend]{algpseudocode}
\usepackage[colorinlistoftodos]{todonotes}

\usepackage{xr}
\usepackage{array}
\usepackage{xspace}

\usepackage{hyperref}
\usepackage[capitalise]{cleveref}
\crefrangelabelformat{enumi}{#3#1#4--#5#2#6}
\usepackage{parskip}
\usepackage{chngcntr}
\usepackage{mathtools}

\usepackage{nameref}

\usepackage{tikz}
\usetikzlibrary{patterns}

\allowdisplaybreaks

\DeclareMathOperator*{\Ber}{Ber}
\DeclareMathOperator*{\argmin}{argmin}

\DeclareMathOperator*{\expectation}{\mathbb{E}}
\let\poly\relax
\DeclareMathOperator*{\poly}{poly}

\DeclareMathOperator*{\probability}{\mathbb{P}}

\renewcommand\R{\mathbb{R}}
\newcommand\Z{\mathbb{Z}}

\newcommand{\expect}{\expectation\expectarg}
\DeclarePairedDelimiterX{\expectarg}[1]{[}{]}{%
	\ifnum\currentgrouptype=16 \else\begingroup\fi
	\activatebar#1
	\ifnum\currentgrouptype=16 \else\endgroup\fi
}

\newcommand{\expectover}[1]{\expectation_{#1}\expectarg}

\newcommand{\innermid}{\nonscript\;\delimsize\vert\nonscript\;}
\newcommand{\activatebar}{%
	\begingroup\lccode`\~=`\|
	\lowercase{\endgroup\let~}\innermid 
	\mathcode`|=\string"8000
}

\newcommand\prob[1]{\probability\left( #1 \right)}
\newcommand\probarg[2]{\probability_{#2}\left( #1 \right)}

\newcommand\opt{\textsc{Opt}\xspace}
\newcommand\copt{c(\textsc{Opt})\xspace}
\newcommand\lpopt{\textsc{LP}_{\textsc{Opt}}}
\newcommand\alg{\textsc{Alg}\xspace}

\counterwithin{equation}{section}

\newcommand{\ForEach}{\textbf{for each}\xspace}

\newcommand\sscText{$\ComplexityFont{SSC}$}
\newcommand\osscText{$\ComplexityFont{FDSSC}$}

\usepackage{eqparbox}

\newcommand\fmin{f_{\min}}

\newcommand\lastt{{t-1}}
\newcommand\thist{{t}}
\newcommand\dthrsh{\gamma}

\newcommand\KL[2]{\textsc{KL}\left(#1 \mid \mid #2\right)}
\newcommand\wKL[2]{\textsc{KL}_c\left(#1 \mid \mid #2\right)}

\newcommand\setcov{\textsc{SetCover}\xspace}

\newcommand\subcov{\textsc{SubmodularCover}\xspace}
\newcommand\osetcov{\textsc{OnlineSetCover}\xspace}
\newcommand\roosc{\textsc{ROSetCover}\xspace}
\newcommand\roocip{\textsc{ROCIP}\xspace}
\newcommand\streamsetcov{\textsc{StreamingSetCover}\xspace}
\newcommand\rosubc{\textsc{ROSubmodularCover}\xspace}
\renewcommand\SAT{\textsc{SAT}\xspace}

\theoremstyle{plain}

\newtheorem{theorem}{Theorem}[section]

\newtheorem{lemma}[theorem]{Lemma}

\newtheorem{claim}[theorem]{Claim}

\newtheorem{corollary}[theorem]{Corollary}

\newtheorem{invariant}{Invariant}

\newlength{\continueindent}
\usepackage{etoolbox}
\makeatletter
\newcommand*{\ALG@customparshape}{\parshape 2 \leftmargin \linewidth \dimexpr\ALG@tlm+\continueindent\relax \dimexpr\linewidth+\leftmargin-\ALG@tlm-\continueindent\relax}
\apptocmd{\ALG@beginblock}{\ALG@customparshape}{}{\errmessage{failed to patch}}
\makeatother

\makeatletter
\def\thm@space@setup{%
	\thm@preskip=\parskip \thm@postskip=0pt
}
\makeatother

\usepackage{etoolbox}
\usepackage{tikz}
\usetikzlibrary{tikzmark}
\usetikzlibrary{calc}

\errorcontextlines\maxdimen

\newcommand{\ALGtikzmarkcolor}{black}
\newcommand{\ALGtikzmarkextraindent}{4pt}
\newcommand{\ALGtikzmarkverticaloffsetstart}{-.5ex}
\newcommand{\ALGtikzmarkverticaloffsetend}{-.5ex}
\makeatletter
\newcounter{ALG@tikzmark@tempcnta}

\newcommand\ALG@tikzmark@start{%
	\global\let\ALG@tikzmark@last\ALG@tikzmark@starttext%
	\expandafter\edef\csname ALG@tikzmark@\theALG@nested\endcsname{\theALG@tikzmark@tempcnta}%
	\tikzmark{ALG@tikzmark@start@\csname ALG@tikzmark@\theALG@nested\endcsname}%
	\addtocounter{ALG@tikzmark@tempcnta}{1}%
}

\def\ALG@tikzmark@starttext{start}
\newcommand\ALG@tikzmark@end{%
	\ifx\ALG@tikzmark@last\ALG@tikzmark@starttext
	\else
	\tikzmark{ALG@tikzmark@end@\csname ALG@tikzmark@\theALG@nested\endcsname}%
	\tikz[overlay,remember picture] \draw[\ALGtikzmarkcolor] let \p{S}=($(pic cs:ALG@tikzmark@start@\csname ALG@tikzmark@\theALG@nested\endcsname)+(\ALGtikzmarkextraindent,\ALGtikzmarkverticaloffsetstart)$), \p{E}=($(pic cs:ALG@tikzmark@end@\csname ALG@tikzmark@\theALG@nested\endcsname)+(\ALGtikzmarkextraindent,\ALGtikzmarkverticaloffsetend)$) in (\x{S},\y{S})--(\x{S},\y{E});%
	\fi
	\gdef\ALG@tikzmark@last{end}%
}

\apptocmd{\ALG@beginblock}{\ALG@tikzmark@start}{}{\errmessage{failed to patch}}
\pretocmd{\ALG@endblock}{\ALG@tikzmark@end}{}{\errmessage{failed to patch}}
\makeatother

\algblock[with]{With}{EndWith}
\algblockdefx[With]{With}{EndWith}%
[1]{\textbf{with} #1 \textbf{do}}%
{}

\makeatletter
\ifthenelse{\equal{\ALG@noend}{t}}%
{\algtext*{EndWith}}
{}%
\makeatother

\title{Random Order Set Cover is as Easy as Offline}
\author{Anupam Gupta\thanks{Computer Science Department, Carnegie Mellon
		University, Pittsburgh, PA 15213. Emails:
		\texttt{\{anupamg,roiel\}@cs.cmu.edu}. Supported in part by NSF awards CCF-1907820, CCF1955785, and CCF-2006953.}
	\and
	Gregory Kehne\thanks{School of Engineering and Applied Sciences, Harvard
		University, Boston, MA 02138. Email:
		\texttt{gkehne@g.harvard.edu}. }
	\and
	Roie Levin$^{*}$
}
\date{}

\begin{document}

	\maketitle
	
	\begin{abstract}
	    We give a polynomial-time algorithm for \osetcov with a competitive ratio of $O(\log mn)$ when the elements are revealed in random order, essentially matching the best possible offline bound of $O(\log n)$ and circumventing the $\Omega(\log m \log n)$ lower bound known in adversarial order. We also extend the result to solving pure covering IPs when constraints arrive in random order.
	    
	    The algorithm is a multiplicative-weights-based round-and-solve approach we call \textsc{LearnOrCover}. We maintain a coarse fractional solution that is neither feasible nor monotone increasing, but can nevertheless be rounded online to achieve the claimed guarantee (in the random order model). This gives a new offline algorithm for \setcov that performs a single pass through the elements, which may be of independent interest.
	\end{abstract}

	\section{Introduction}
	
	In the \setcov problem we are given a set system $(U,\mathcal{S})$ (where $U$ is a ground set of size $n$ and $\mathcal{S}$ is a collection of subsets with $|\mathcal{S}| = m$), along with a cost function $c:
        \mathcal{S}\rightarrow \R^+$. The goal is to select a minimum
        cost subcollection $\mathcal{S}' \subseteq \mathcal{S}$ such
        that the union of the sets in $\mathcal{S}'$ is $U$. Many
        algorithms have been discovered for this problem that achieve
        an approximation ratio of $\ln n$ (see
        e.g. \cite{chvatal1979greedy,johnson1974approximation,lovasz1975ratio,williamson2011design}),
        and this is best possible unless $\P = \NP$
        \cite{feige1998threshold,Dinur:2014:AAP:2591796.2591884}.
	
	In the \osetcov variant, we impose the additional restriction that the algorithm does not know $U$ initially, nor the contents of each $S \in \mathcal{S}$. Instead, an adversary reveals the elements of $U$ one-by-one in an arbitrary order. On the arrival of every element, it is revealed which sets of $S \in \mathcal{S}$ contain the element, and the algorithm must immediately pick one such set $S \in \mathcal{S}$ to cover it. The goal is to minimize the total cost of the sets chosen by the algorithm. In their seminal work, \cite{alon2003online} show that despite the lack of foresight, it is possible to achieve a competitive ratio of $O(\log m \log n)$ for this version\footnote{Throughout this paper we consider the \emph{unknown-instance model} for online set cover (see Chapter 1 of \cite{korman2004use}). The result of \cite{alon2003online} was presented for the \emph{known-instance model}, but extends to the unknown setting as well. This was made explicit in subsequent work \cite{buchbinder2009online}.}. This result has since been shown to be tight unless $\NP \subseteq \BPP$ \cite{korman2004use}, and has also been generalized significantly (e.g. \cite{alon2006general,buchbinder2009online,gupta2014approximating,gupta2020online}).
	
	In this paper, we answer the question:
	\begin{quote}
		\emph{What is the best competitive ratio possible for \osetcov when the adversary is constrained to reveal the elements in uniformly random order (RO)?}
	\end{quote}

	We call this version \roosc. 
	Note that the element set $U$ is still adversarially chosen and unknown, and only the arrival order is random. 
	
	\subsection{Results}
	
    We show that with only this one additional assumption on the element arrival order, there is an efficient algorithm for \roosc with expected competitive ratio matching the best-possible offline approximation guarantee (at least in the regime where $m = \poly(n)$).
	
	\begin{theorem}
          \nameref{alg:gencost} is a polynomial-time randomized algorithm for \roosc achieving expected competitive ratio $O(\log (mn))$.
	\end{theorem}

	When run offline, our approach gives a new asymptotically optimal algorithm for \setcov for $m = \poly(n)$, which may be of independent interest. Indeed, given an estimate for the optimal \emph{value} of the set cover, our algorithm makes a single pass over the elements (considered in random order), updating a fractional solution using a multiplicative-weights framework, and sampling sets as it goes. This simplicity, and the fact that it uses only $\tilde O(m)$ bits of memory, may make the algorithm useful for some low-space streaming applications.  (Note that previous formulations of \streamsetcov \cite{saha2009maximum,demaine2014streaming,har2016towards,emek2016semi} only consider cases where \emph{sets} arrive in a stream.)

	We show next that a suitable generalization of the same algorithm achieves the same competitive ratio for the \emph{RO Covering Integer Program} problem (\roocip) (see \cref{sec:cips} for a formal description).
	\begin{theorem}
		\label{thm:intro_cip}
		\nameref{alg:cips} is a polynomial-time randomized algorithm for \roocip achieving competitive ratio $O(\log (mn))$.
	\end{theorem}
	
	We complement our main theorem with some lower bounds. For instance, we show that the algorithms of \cite{alon2003online,buchbinder2009online} have a performance of $\Theta(\log m \log n)$ even in RO, so a new algorithm is indeed needed. Moreover, we observe an $\Omega(\log n)$ lower bound on \emph{fractional} algorithms for \roosc. This means we cannot pursue a two-phase strategy of maintaining a good monotone fractional solution and then randomly rounding it (as was done in prior works) without losing $\Omega(\log^2 n)$. Interestingly, our algorithm \emph{does} maintain a (non-monotone) fractional solution while rounding it online, but does so in a way that avoids extra losses. We hope that our approach will be useful in other works for online problems in RO settings. (We also give other lower bounds for batched versions of the problem, and for the more general submodular cover problem.)

	\subsection{Techniques and Overview}
	
	The core contribution of this work is demonstrating that one can exploit randomness in the arrival order to \textit{learn} about the underlying set system. 
	What is more, this learning can be done fast enough (in terms
        of both sample and computational complexity) to build an
        $O(\log mn)$-competitive solution, even while committing to
        covering incoming elements immediately upon arrival. This
        seems like an idea with applications to other
        sequential decision-making problems, particularly in the RO
        setting.
        
        We start in \cref{sec:exptime} with an exponential time algorithm for the unit-cost setting. 
	This algorithm maintains a portfolio of all exponentially-many collections of cost $c(\opt)$ that are feasible for the elements observed so far. 
	When an uncovered element arrives, the algorithm takes a
        random collection from the portfolio, and picks a random set
        from it covering the element. It then prunes the portfolio to drop collections that did not cover the incoming element. 
	We show that either the expected marginal coverage of the chosen set is large, or the expected number of solutions removed from the portfolio is large. 
	I.e., we either make progress \textit{covering}, or \textit{learning}. 
	We show that within $c(\opt) \log (mn)$ rounds, the portfolio
        only contains the true optimal feasible solutions, or all
        unseen elements are covered.
        (One insight that comes from our result is that a good measure of set quality is the number of times an unseen \textit{and uncovered} element appears in it.)
	
        We then give a polynomial-time algorithm in \cref{sec:gen_cost}:
        while it is quite different from the exponential scheme above,
        it is also based on this insight, and the intuition that our algorithm should make
        progress via  learning or covering. 
	Specifically, we maintain a distribution $\{ x_S \}_{S \in
          \mathcal{S}}$ on sets: for each arriving uncovered element,
        we first sample from this distribution $x$, then update $x$ via a multiplicative weights rule. 
	If the element remains uncovered, we buy the cheapest set covering it. 
	For the analysis, we introduce a potential function which
        simultaneously measures the convergence of this distribution
        to the optimal fractional solution, and the progress towards
        covering the universe. 
	Crucially, this progress is measured in expectation over the
        random order, thereby circumventing 
        lower bounds for the adversarial-order setting
        \cite{korman2004use}.
        In \cref{sec:cips} we extend our method to the more general
        \roocip problem: the intuitions and general proof outlines are
        similar, but we need to extend the algorithm to handle elements
        being partially covered. 
	
        Finally, we present lower bounds in \cref{sec:lbs}.
		Our information-theoretic lower bounds for \roosc follow from elementary combinatorial arguments. We also show a lower bound for a batched version of \roosc, following the hardness proof of \cite{korman2004use}; we use it in turn to derive lower bounds for \rosubc.

	\subsection{Other Related Work}
	
	There has been much recent interest in algorithms for random order online problems.
	Starting from the secretary problem, 
	the RO model has been extended to  include metric facility
        location \cite{meyerson2001online}, network design
        \cite{meyerson2001designing}, and solving packing LPs
        \cite{AWY14,MR12,kesselheim2014primal,GM16,AgrawalD15,albers2021improved},
        load-balancing \cite{GM16,Molinaro-SODA17} and scheduling
        \cite{albers2020scheduling,AJ21}.
	See \cite{gupta2020random} for a recent survey.
	
	Our work is closely related to \cite{grandoni2008set}, who
        give an $O(\log mn)$-competitive algorithm 
        a related stochastic model, where the elements are chosen
        i.i.d.\ from a \emph{known distribution} over the elements. 
	\cite{dehghani2018greedy} 
        generalize the result of
        \cite{grandoni2008set} to the \textit{prophet
          version}, in which elements are drawn from known but
        distinct distributions. 
	Our work is a substantial strengthening, since the RO model
        captures the \emph{unknown} i.i.d.\ setting as a special
        case. Moreover, the learning is an important and interesting
        part of our technical contribution.  
	On the flip side, the algorithm of \cite{grandoni2008set} satisfies
        \textit{universality} (see their work for a definition), which our algorithm does not. 
        We point out that \cite{grandoni2008set} claim (in a note, without proof) that it is not possible to circumvent the $\Omega (\log m \log n)$ lower bound of \cite{korman2004use} even in RO; our results show this claim is incorrect. 
	We discuss this in \cref{sec:ggl}.

        Regret bounds for online learning are also proven via the KL
        divergence (see, e.g.,
        \cite{arora2012multiplicative}). However, 
        no reductions are known from our problem to the classic MW
        setting, and it is unclear how random order would play a role
        in the analysis: this is necessary to bypass the adversarial order lower bounds. 

        Finally, \cite{korman2004use} gives an algorithm with
        competitive ratio $k \log (m/k)$---hence total cost $k^2
        \log (m/k)$---for unweighted \osetcov where $k = |\opt|$. The
        algorithm is the same as our exponential time algorithm in
        \cref{sec:exptime} for the special case of $k=1$; however, we
        outperform it for non-constant $k$, and generalize it to get
        polynomial-time algorithms as well.
	
	\subsection{Preliminaries}
	
	All logarithms in this paper are taken to be base $e$. In the following definitions, let $x,y \in \R^n_+$ be vectors. 
        The standard dot product between $x$ and $y$ is denoted
	$\langle x, y\rangle = \sum_{i=1}^n x_i y_i$. We use a weighted generalization of KL divergence. Given
        a weight function $c$, define
	\begin{align}\wKL{x}{y} : = \sum_{i=1}^n c_i \left[x_i \log \left(\frac{x_i}{y_i}\right) -x_i + y_i\right].\label{eq:prelim_kl}\end{align}
	The weighted KL divergence is a nonnegative quantity, which can be demonstrated via the inequality $a - b \leq a \log \frac{a}{b}$ for $a, b \geq 0$, also known as the ``Poor man's Pinsker'' inequality.

	\section{Warmup: An Exponential Time Algorithm for Unit Costs}
	\label{sec:exptime}

	We begin with an exponential-time algorithm which we call \textsc{SimpleLearnOrCover} to demonstrate some core ideas of our result. In what follows we assume that we know a number $k \in [\copt, 2\cdot \copt]$. This assumption is easily removed by a standard guess and double procedure, at the cost of an additional factor of $2$.
	
	The algorithm is as follows. Maintain a list $\mathfrak{T} \subseteq \binom{\mathcal{S}}{k}$ of candidate $k$-tuples of sets. When an uncovered element $v$ arrives, choose a $k$-tuple $\mathcal{T} = (T_1, \ldots, T_k)$ uniformly at random from $\mathfrak{T}$, and buy a uniformly random $T$ from $\mathcal{T}$. Also buy an arbitrary set containing $v$. Finally, discard from $\mathfrak{T}$ any $k$-tuples that do not cover $v$. See \cref{sec:appendix2} for pseudocode.
	
	\begin{theorem}
	    \label{thm:unit_cost_exp}
	    \nameref{alg:expunit} is a randomized algorithm for unit-cost \roosc with expected cost $O(k  \cdot \log (mn))$.
	\end{theorem}
	
	\begin{proof}[Proof of \cref{thm:unit_cost_exp}]
        Consider any time step $\thist$ in which a random element arrives that is \textit{uncovered} on arrival. Let $U^{\thist}$ be the set of elements that remain uncovered at the end of time step $\thist$. Before the algorithm takes action, there are two cases:
		
        \textbf{Case 1:} At least half the tuples in $\mathfrak{T}$ cover at least $|U^\lastt|/2$ of the $|U^\lastt|$ as-of-yet-uncovered elements. In this case we say the algorithm performs a \textbf{Cover Step} in round $t$.
		
        \textbf{Case 2:} At least half the tuples in $\mathfrak{T}$ cover strictly less than $|U^\lastt|/2$ of the uncovered elements; we say the algorithm performs a \textbf{Learning Step} in round $t$.
		
        Define $\mathfrak{N}(c)$ to be the number of uncovered elements remaining after $c$ Cover Steps. Define $\mathfrak{M}(\ell)$ to be the value of $|\mathfrak{T}|$ after $\ell$ Learning Steps. We will show that after $10 k (\log m + \log n)$ rounds, either the number of elements remaining is less than $\mathfrak{N}(10 k \log n)$ or the number of tuples remaining is less than $\mathfrak{M}(10 k \log m)$. In particular, we argue that both $\expect*{\mathfrak{N}(10 k \log n)}$ and $\expect*{\mathfrak{M}(10 k \log m)}$ are less than $1$. 
		
		\begin{claim}
			$\expect*{\mathfrak{N}(c+1) | \mathfrak{N}(c) = N} \leq \left(1-\frac{1}{4k}\right) N$.
		\end{claim}
		\begin{proof}
			If round $t$ is a Cover Step, then at least half of the $\mathcal{T} \in \mathfrak{T}$ cover at least half of $U^{\lastt}$, so $\expect*{\left\lvert(\bigcup \mathcal{T}) \cap U^\lastt\right\rvert} \geq \frac{|U^\lastt|}{4}$. Since $T$ is drawn uniformly at random from the $k$ sets in the uniformly random $\mathcal{T}$, we have $\expect*{\left\lvert T \cap U^\lastt\right\rvert} \geq \frac{|U^\lastt|}{4k}$. 
		\end{proof}
	
		\begin{claim}
			$\expect*{\mathfrak{M}(\ell+1) \mid \mathfrak{M}(\ell) = M} \leq \frac{3}{4} M$.
		\end{claim}
		\begin{proof}
			Upon the arrival of $v$ in a Learning Step, at least half the tuples have probability at least $\nicefrac{1}{2}$ of being removed from $\mathfrak{T}$, so the expected number of tuples removed from $\mathfrak{T}$ is at least $M/4$.
		\end{proof}
		
		To conclude, by induction 
		\begin{align*}
			\textstyle \expect*{\mathfrak{N}(10 k \log n)} 
			\leq n\,\left(1-\frac{1}{4k}\right)^{10 k\log n}
			\leq 1 \quad \text{and} \quad 
			\expect*{\mathfrak{M}(10 k \log m)} 
			\leq m^k \, \left(\frac{3}{4}\right)^{10 k\log m} \leq 1.
		\end{align*}
		Note that if there are $N$ remaining uncovered elements and $M$ remaining tuples to choose from, the algorithm will pay at most $\min(k\cdot M,N)$ before all elements are covered. Thus the total expected cost of the algorithm is bounded by
		\[10k (\log m + \log n) + \expect*{\min(k\cdot \mathfrak{M}(10 k \log m), \mathfrak{N}(10 k \log n))} = O(k \log mn). \qedhere\]
		
	\end{proof}

	Apart from the obvious challenge of modifying \textsc{SimpleLearnOrCover} to run in polynomial time, it is unclear how to generalize it to handle non-unit costs. Still, the intuition from this algorithm will be useful for the next sections.

		\section{A Polynomial-Time Algorithm for General Costs}

	\label{sec:gen_cost}

	We build on our intuition from \cref{sec:exptime} that we can
        either make progress in covering or in learning about the optimal solution.
	To get an efficient algorithm, we directly
	maintain a probability distribution over sets, which we update via a
	multiplicative weights rule. We use a potential function that
        simultaneously measures progress towards learning the optimal
        solution, and towards covering the unseen elements. Before we
          present the formal details and the pseudocode, here are the main pieces of the algorithm.

          \begin{enumerate}
          \item We maintain a fractional vector $x$ which
            is a (\emph{not necessarily feasible}) guess for the LP
            solution of cost $\beta$ to the set cover instance. 
          \item Every
            round $\thist$ in which an uncovered element $v^\thist$ arrives, we
            \begin{enumerate}
            \item sample every set $S$ with probability
              proportional to its current LP value $x_S$,
            \item  increase
              the value $x_S$ of all sets $S \ni v^\thist$ multiplicatively
              and renormalize,
            \item buy a cheapest set to
              cover $v^\thist$ if it remains uncovered.
            \end{enumerate}
          \end{enumerate}
       
          \bigskip
          Formally, by a guess-and-double approach, we assume we know
          a bound $\beta$ such that
          $\lpopt \leq \beta \leq 2 \cdot \lpopt$; here $\lpopt$ is
          the cost of the optimal LP solution to the final unknown
          instance.  Define
        \begin{gather}
          \kappa_v := \min \{c_S \mid S \ni v\}
        \end{gather}
          as the cost of the cheapest set covering $v$.  
	\begin{algorithm}[H]
	\caption{\textsc{LearnOrCover}}
	\label{alg:gencost}
	\begin{algorithmic}[1]
		\State Initialize $\mathcal{C}^0 \leftarrow \{S : c(S) < \beta/m \}$. \label{line:loc_cheapsets}
		\State Initialize $x_S^{0} \leftarrow \frac{\beta}{c_S \cdot m'} \cdot \mathbbm{1}\{\beta/m \leq c(S) \leq \beta \}$, where $m' := |\{S \ : \  \beta/m \leq c(S) \leq \beta \}|$. \label{line:loc_sup}
		\For{$\thist=1,2\ldots, n$}
		\State{$v^\thist \leftarrow$ $\thist^{th}$ element in the
                  random order, and let $\mathcal{R}^\thist \leftarrow \emptyset$.}
        \If {$v^\thist$ not already covered} \label{line:gencost_uncov}
			\State \ForEach set $S$, with probability $\min(\kappa_{v^\thist} \cdot x^{\lastt}_S/\beta,1)$ add $\mathcal{R}^\thist \leftarrow \mathcal{R}^\thist \cup \{S\}$.
			\State Update $\mathcal{C}^\thist \leftarrow \mathcal{C}^{\lastt} \cup \mathcal{R}^\thist$.
			\If {$\sum_{S \ni v^\thist} x^{\lastt}_S < 1$}
			\State For every set $S$, update $x^{\thist}_S \leftarrow x^{\lastt}_S \cdot \exp\left\{\mathbbm{1}\{S \ni v^\thist \} \cdot \kappa_{v^{\thist}} / c_S\right\}$. 
			\State Let $Z^{\thist} = \langle c, x^\thist \rangle / \beta$ and normalize $x^{\thist} \leftarrow x^{\thist} / Z^{\thist}$. \label{line:loc_cost_invariant}
			\Else
			\State $x^\thist \leftarrow x^{\lastt}$. \label{line:gencost_noupdate}
			\EndIf
			\State Let $S_{v^\thist}$ be the cheapest set containing $v^\thist$. Add $\mathcal{C}^\thist \leftarrow \mathcal{C}^\thist \cup \{S_{v^\thist}\}$. \label{line:gencost_backup}
		\EndIf
		\EndFor 
	\end{algorithmic}
\end{algorithm}

The algorithm is somewhat simpler for unit costs: the $x_S$ values are
multiplied by either $1$ or $e$, and moreover we can sample a single
set for $\mathcal{R}^\thist$ (see \cref{sec:appendix2} for pseudocode). Because of the non-uniform set costs, we have to carefully calibrate both the learning and sampling rates. Our algorithm dynamically scales the learning and sampling rates in round $\thist$ depending on $\kappa_{v^\thist}$, the cost of the cheapest set covering $v^\thist$. Intuitively, this ensures that all three of (a) the change in potential, (b) the cost of the sampling, and (c) the cost of the backup set, are at the same scale. Before we begin, observe that \cref{line:loc_cost_invariant} ensures the following invariant:

\begin{invariant}
	\label{inv:loc_cost}
	For all time steps $\thist$, it holds that $\langle c, x^\thist \rangle = \beta$.
\end{invariant}

\begin{theorem}[Main Theorem]
	\label{thm:weighted_cost_poly}
	\nameref{alg:gencost} is a polynomial-time randomized algorithm for \roosc with expected cost $O(\beta  \cdot \log (mn))$.
\end{theorem}

	 Let us start by defining notation. 
	 Let $x^*$ be the optimal LP solution to the final, unknown set cover instance. 
	 Next let $X_v^{\thist} := \sum_{S \ni v} x_S^{\thist}$, the fractional coverage provided to $v$ by $x^{\thist}$. Let $U^{\thist}$ be the elements remaining uncovered at the end of round $t$ (where $U^{0}=U$ is the entire ground set of the set system). Define the quantity 
	 \begin{align}
	 	\label{eq:gencost_rhodef}
	 	\rho^{\thist} :=  \sum_{u \in U^{\thist}} \kappa_u.
 	\end{align}
	 With this, we are ready to define our potential function which is the central player in our analysis. Recalling that $\beta$ is our guess for the value of $\lpopt$, and the definition of KL divergence \eqref{eq:prelim_kl}, define
		\begin{align}
			\label{eq:gencost_phi}
			\boxed{
			\Phi(\thist) := C_1 \cdot \wKL{x^*}{x^{\thist}} + C_2 \cdot \beta \cdot  \log \left(\frac{\rho^{\thist}}{\beta} + \frac{1}{m}\right)}
		\end{align}
	where $C_1$ and $C_2$ are constants to be specified later.

	\begin{lemma}[Potential Bounds]
		\label{fact:gencost_phi0}
		The initial potential is bounded as $\Phi(0) = O(\beta\cdot \log (mn))$, and $\Phi(t) \geq -O(\beta\cdot \log m )$ for all rounds $t$.
	\end{lemma}
	We write the proof in general language in order to reuse it for covering IPs in the following section. Recall that sets correspond to columns in the canonical formulation of \setcov as an integer program.
	\begin{proof}
		We start with the initial potential. We require the following fact which we prove in \cref{sec:appendix}.
		\begin{restatable}{fact}{lpsupport}
			\label{fact:lpsupport}
			Every pure covering LP of the form $\min_{x \geq 0}\{\langle c, x \rangle : Ax \geq 1\}$ for $c \geq 0$ and $a_{ij} \in [0, 1]$ with optimal value less than $\beta$ has an optimal solution $x^*$ which is supported only on columns $j$ such that $c_j \leq \beta$.
		\end{restatable}
		We assume WLOG that $x^*$ is such a solution, and we first bound the KL term of $\Phi(0)$. Since
                $\text{support}(x^*) \subseteq \text{support}(x^0)$ by \cref{fact:lpsupport}, we have
                \[\wKL{x^*}{x^{0}} = \sum_{j} c_j \cdot  x^*_j \log
                \left(x^*_j \frac{c_j \cdot m'}{\beta}\right) + \sum_{j} c_j (x^0_j - x^*_j)\leq
                \beta (\log (m)+1),\]
		 where we used that $\langle c,
                 x^* \rangle \leq \beta$ and that $m'\leq m$ is the number of columns with cost in $[\beta/m, \beta]$. 
		
		For the second term,  $\beta \log (\rho^{0}/\beta + 1/m) = \beta \log (\sum_{i \in U^0} \kappa_i /\beta + 1/m)
		\leq \beta \log (|U^{0}| + 1/m) \leq \beta \log (n+1)$, since for
		all $i$, the cheapest cover for $i$ costs less than
		$\beta$, and therefore $\kappa_i / \beta \leq 1$. The upper bound follows so long as $C_1$ and $C_2$ are constants.
		
		We conclude with the lower bound. The weighted KL term is nonnegative, and the second term is bounded below by $C_2 \beta \cdot \log(1/m) = - C_2 \beta \log m$, implying the claim.
	\end{proof}
	
	The rest of the proof relates the expected decrease of
        potential $\Phi$ to the algorithm's cost in each round. 
	Define the event $\Upsilon^{\thist} := \{v^\thist \in U^{\lastt}\}$ that the element $v^{\thist}$ is \emph{uncovered} on arrival. 
	Note \cref{line:gencost_uncov} ensures that if event $\Upsilon^{\thist}$ does not hold, the algorithm
        takes no action and the potential does not change.
        So we focus on the case that  event $\Upsilon^{\thist}$ does occur. 
	We first analyze the change in KL divergence. Recall that $X_v^{\thist} := \sum_{S \ni v} x_S^{\thist}$. 
	
	\begin{lemma}[Change in KL]
	\label{lem:weighted_klchange}
	    For rounds $\thist$ in which $\Upsilon^{\thist}$ holds, the expected change in the weighted KL divergence is
		\[
			\expectover{v^\thist, \mathcal{R}^{\thist}}*{ \wKL{x^*}{x^{\thist}} - \wKL{x^*}{x^{\lastt}} \mid x^{\lastt}, U^{\lastt}, \Upsilon^{\thist}} \leq \expectover{v \sim U^{\lastt}}{ (e-1) \cdot\kappa_v \cdot \min(X_v^{\lastt},1) - \kappa_v}.
		\]
	\end{lemma}
	We emphasize that the expected change in relative entropy in the statement above depends only on the randomness of the arriving uncovered element $v^\thist$, not on the randomly chosen sets $\mathcal{R}^{\thist}$.

	\begin{proof} We break the proof into cases. If $X_v^{\lastt} \geq 1$, in \cref{line:gencost_noupdate} we set the vector $x^{\thist} = x^{\lastt}$, so the change in KL divergence is 0. This means that 
	\begin{align} &\expectover{v^\thist, \mathcal{R}^{\thist}}*{ \wKL{x^*}{x^{\thist}} - \wKL{x^*}{x^{\lastt}} \mid x^{\lastt}, U^{\lastt}, \Upsilon^{\thist}, X_{v^{\thist}}^{\lastt}  \geq 1} \notag \\
	&\leq \expectover{v \sim U^{\thist}}{(e-1) \cdot \kappa_v \min\left(X_v^{\lastt},1\right) - \kappa_v | X_v^{\lastt} \geq  1} \label{eq:wkl_xvge1}\end{align}
	 trivially. Henceforth we focus on the case $X_{v^{\thist}}^{\lastt} < 1$.
	
	Recall that the expected change in relative entropy depends only on the arriving uncovered element $v^{\thist}$. Expanding definitions,
	\begin{align}
	        & \expectover{v^{\thist}, \mathcal{R}^{\thist}}*{\wKL{x^*}{x^{\thist}} - \wKL{x^*}{x^{\lastt}} \mid x^{\lastt}, U^{\lastt}, \Upsilon^{\thist}, X_{v^{\thist}}^{\lastt} < 1} \notag \\
	        &= \expectover{v \sim U^{\lastt}}*{\sum_{S} c_S \cdot x^*_S \cdot \log \frac{x^{\lastt}_S}{x_S^{\thist}} | X_v^{\lastt} < 1} \notag  \\
	        &= \expectover{v \sim U^{\lastt}}*{\langle c, x^*\rangle \cdot  \log Z^{\thist} - \sum_{S \ni v} c_S \cdot x^*_S \cdot \log e^{\kappa_v/c_S} | X_v^{\lastt} < 1} \notag  \\
	        &\leq \expectover{v \sim U^{\lastt}}*{ \beta \cdot \log\left(\sum_{S\ni v} \frac{c_S}{\beta} \cdot x_S^{\lastt}  \cdot e^{\kappa_v/c_S} + \sum_{S\not \ni v} \frac{c_S}{\beta} \cdot x_S^{\lastt} \right) - \sum_{S \ni v} \kappa_v \cdot x^*_S | X_v^{\lastt} < 1}\label{eq:wkl_copt}, \\
        \intertext{where in the last step \eqref{eq:wkl_copt} we expanded the definition of $Z^{\thist}$, and used $\langle c, x^*\rangle \leq \beta$. Since $x^*$ is a feasible set cover, which means that $\sum_{S \ni v} x^*_S \geq 1$, we can further bound \eqref{eq:wkl_copt} by}
			&\leq \expectover{v \sim U^{\lastt}}*{ \beta \cdot \log\left(\sum_{S\ni v} \frac{c_S}{\beta} \cdot x_S^{\lastt} \cdot e^{\kappa_v/c_S} + \sum_{S\not \ni v} \frac{c_S}{\beta} \cdot x_S^{\lastt} \right) - \kappa_v | X_v^{\lastt} < 1} \notag	
			\\
	        &\leq  \expectover{v \sim U^{\lastt}}*{ \beta  \cdot \log\left(\sum_S \frac{c_S}{\beta} \cdot  x_S^{\lastt} + (e - 1) \cdot \sum_{S\ni v} \frac{\kappa_v}{\beta} \cdot x_S^{\lastt} \right) - \kappa_v | X_v^{\lastt} < 1} \label{eq:wkl_eapx} \\
	        \intertext{where we use the approximation $e^y \leq 1+(e-1) \cdot y$ for $y \in [0,1]$ (note that $\kappa_v$ is the cheapest set covering $v$, so for any $S \ni v$ we have $\kappa_v / c_S \leq 1$). Finally, using \cref{inv:loc_cost}, along with the approximation $\log(1+y) \leq y$, we bound \eqref{eq:wkl_eapx} by}
		  	&\leq  \expectover{v \sim U^{\lastt}}*{(e-1) \cdot \kappa_v \cdot X_v^{\lastt} - \kappa_v | X_v^{\lastt} < 1} \notag 
		  	\\
		  	&\leq  \expectover{v \sim U^{\lastt}}{(e-1) \cdot \kappa_v \min(X_v^{\lastt},1) - \kappa_v | X_v^{\lastt} < 1}. \label{eq:wkl_xvl1}
	    \end{align}
		The lemma statement follows by combining \eqref{eq:wkl_xvge1} and \eqref{eq:wkl_xvl1} using the law of total expectation.
	\end{proof}
	
	Next we bound the expected change in $\log (\rho^{\thist}/\beta + 1/m)$ provided by the sampling $\mathcal{R}^{\thist}\sim \kappa_{v^{\thist}} x^{\lastt} / \beta$ upon the arrival of uncovered $v^{\thist}$ (where each $S \in \mathcal{R}^{\thist}$ independently with probability $\kappa_{v^{\thist}} x_S^{\lastt}/\beta$).
	Recall that $U^{\thist}$ denotes the elements uncovered at the end of round $\thist$; therefore element $u$ is contained in $U^{\lastt} \setminus U^{\thist}$ if and only if it is marginally covered by $\mathcal{R}^{\thist}$ in this round.
	
	\begin{lemma}[Change in $\log\left( \rho^{\thist}/\beta + 1/m\right)$]
		\label{lem:weighted_utchange}
		For rounds when $v$ is uncovered on arrival, the expected change in $\log \left(\rho^{\thist}/\beta + 1/m\right)$ is
		\[
			\expectover{v^{\thist}, \mathcal{R}^{\thist}}*{\log \left(\frac{\rho^{\thist}}{\beta} + \frac{1}{m}\right) - \log \left(\frac{\rho^{\lastt}}{\beta} + \frac{1}{m}\right) \mid x^{\lastt}, U^{\lastt}, \Upsilon^{\thist}} \leq - \frac{1-e^{-1}}{2\beta}\cdot \expectover{u \sim U^{\lastt}}*{\kappa_u \cdot \min(X_u^{\lastt}, 1)}.
		\]
	\end{lemma}
	\begin{proof}	
	Recall the definition of $\rho^{\thist}$ from \eqref{eq:gencost_rhodef}. Conditioned on $v^{\thist} = v$ for any fixed element $v$, the expected change in $\log \rho^{\thist}$ depends only on $\mathcal{R}^{\thist}$. Recall $\mathcal{R}^{\thist}$ is formed by sampling each set $S$ with probability $\kappa_{v} x_S^{\lastt} / \beta$. The expected change is 
	\begin{align}
	        &\expectover{\mathcal{R}^{\thist}}*{\log \left(\frac{\rho^{\thist}}{\beta} + \frac{1}{m}\right) - \log \left(\frac{\rho^{\lastt}}{\beta} + \frac{1}{m}\right) \mid x^{\lastt}, U^{\lastt}, \Upsilon^{\thist}, v^{\thist} = v} \notag \\
	        &= \expectover{\mathcal{R}^{\thist}}*{\log \left(1-\frac{\rho^{\lastt} - \rho^{\thist}}{\rho^{\lastt} + \frac{\beta}{m}} \right) | U^{\lastt}, v^{\thist} = v} \notag \\
	        &\leq - \frac{1}{\rho^{\lastt} + \frac{\beta}{m}} \cdot \expectover{\mathcal{R}^{\thist}}*{\rho^{\lastt} - \rho^{\thist} | U^{\lastt}, v^{\thist} = v}. \label{eq:ut_logapx}
			\intertext{Above, \eqref{eq:ut_logapx} follows from the approximation $\log (1-y) \leq -y$. Expanding the definition of $\rho^{\thist}$ from \eqref{eq:gencost_rhodef}, \eqref{eq:ut_logapx} is bounded by}
	        &= - \frac{1}{\rho^{\lastt} + \frac{\beta}{m}} \cdot \expectover{\mathcal{R}^{\thist}}*{\sum_{u \in U^{\lastt}} \kappa_u \cdot \mathbbm{1}\{u \in U^{\lastt}\setminus U^{\thist}\}  | U^{\lastt}, v^{\thist} = v} \notag \\
	        &= - \frac{1}{\rho^{\lastt} + \frac{\beta}{m}} \cdot \sum_{u \in U^{\lastt}} \kappa_u \cdot \probarg{u \not \in U^{\thist} \mid u \in U^{\lastt}, v^{\thist} = v}{\mathcal{R}^{\thist}} \notag \\
	        &\leq - \frac{1-e^{-1}}{\beta} \cdot \kappa_v \cdot \frac{1}{\rho^{\lastt} + \frac{\beta}{m}} \cdot \sum_{u \in U^{\lastt}} \kappa_u  \cdot \min(X_u^{\lastt}, 1) \label{eq:ut_prob}\\
	        &= -  \frac{1-e^{-1}}{\beta}\cdot \kappa_v \cdot \frac{|U^{\lastt}|}{\rho^{\lastt} + \frac{\beta}{m}} \cdot \expectover{u \sim U^{\lastt}}*{\kappa_u \cdot \min(X_u^{\lastt}, 1)} \label{eq:ut_simplify} .
			\intertext{Step \eqref{eq:ut_prob} is due to the fact that each set $S\in \mathcal{R}^{\thist}$ is sampled independently with probability $\min(\kappa_v x^{\lastt}_S /\beta,1)$, so the probability any given element $u \in U^{\lastt}$ is covered is 
		    \[
    			1 - \prod_{S \ni u} \left(1 - \min\left(\frac{\kappa_v x^{\lastt}_S}{\beta},1\right)\right) \geq 1 - \exp\left\{-\min\left(\frac{\kappa_v}{\beta} X_u^{\lastt},1\right)\right\} 
		    	\stackrel{(**)}{\geq} (1 - e^{-1}) \cdot \min\left(\frac{\kappa_v}{\beta} X_u^{\lastt},1\right).
		    \] 
		   	Above, $(**)$ follows from convexity of the exponential; step \eqref{eq:ut_prob} then follows since $\kappa_v / \beta \leq 1$. Taking the expectation of \eqref{eq:ut_simplify} over $v^{\thist} \sim U^{\lastt}$, and using the fact that $\expectover{v \sim U^{\lastt}}*{\kappa_v} = \rho^{\lastt} / |U^{\lastt}|$, the expected change becomes}
    	    &\expectover{v^{\thist}, \mathcal{R}^{\thist}}*{\log \left(\frac{\rho^{\thist}}{\beta} + \frac{1}{m}\right) - \log \left(\frac{\rho^{\lastt}}{\beta} + \frac{1}{m}\right) \mid x^{\lastt}, U^{\lastt}, \Upsilon^{\thist}} \notag \\
			&\leq -  \frac{1-e^{-1}}{\beta} \cdot \frac{\rho^{\lastt}}{\rho^{\lastt} + \frac{\beta}{m}} \cdot \expectover{u \sim U^{\lastt}}*{\kappa_u \cdot \min(X_u^{\lastt}, 1)} \notag \\
    		&\leq -  \frac{1-e^{-1}}{\beta} \cdot \frac{1}{2} \cdot \expectover{u \sim U^{\lastt}}*{\kappa_u \cdot \min(X_u^{\lastt}, 1)}, \label{eq:ut_lastline}
    \end{align}    
     where \eqref{eq:ut_lastline} follows since $\Upsilon^{\thist}$ holds, so $U^{\lastt}$ is nonempty; if there are uncovered elements then $\rho^{\lastt} >0$, and so $\rho^{\lastt} \geq \beta/m$ by \cref{line:loc_cheapsets}. \qedhere
	\end{proof}
	
	\begin{proof}[Proof of \cref{thm:weighted_cost_poly}]
	The algorithm starts by buying cheap sets $\mathcal{C}^0$ on \cref{line:loc_cheapsets} which cost at most $\beta$ in total.
	
	In every round $\thist$ for which $\Upsilon^\thist$ holds, the expected cost of the sampled sets $\mathcal{R}^\thist$ is $\kappa_{v^\thist} \cdot \langle c, x^{\lastt}\rangle / \beta = \kappa_{v^\thist}$ (by \cref{inv:loc_cost}). The algorithm pays an additional $\kappa_{v^\thist}$ in \cref{line:gencost_backup}, and hence the total expected cost per round is at most $2\cdot \kappa_{v^\thist}$.
	
	Combining \cref{lem:weighted_klchange,lem:weighted_utchange}, and setting the constants $C_1 = 2$ and $C_2 = 4e$, we have
	\begin{align*}
		&\expectover{\substack{v^{\thist}, \mathcal{R}^{\thist}}}*{\Phi(\thist) - \Phi(\lastt) | v^{1}, \ldots, v^{\lastt}, \mathcal{R}^{1}, \ldots, \mathcal{R}^{\lastt}, \Upsilon^{\thist}} \\
		&= \expectover{\substack{v^{\thist}, \mathcal{R}^{\thist}}}*{
			\begin{array}{ll}
				&C_1\left(\wKL{x^*}{x^{\thist}} - \wKL{x^*}{x^{\lastt}}\right) \\
				+ &C_2\cdot \beta \cdot \left(\log \left(\frac{\rho^{\thist}}{\beta} + \frac{1}{m}\right) - \log \left(\frac{\rho^{\lastt}}{\beta} + \frac{1}{m}\right)\right)
			\end{array}
		| v^{1}, \ldots, v^{\lastt}, \mathcal{R}^{1}, \ldots, \mathcal{R}^{\lastt}, \Upsilon^{\thist}} \\
		&\leq - \expectover{\substack{v^{\thist}, \mathcal{R}^{\thist}}}*{2 \cdot \kappa_{v^{\thist}} | v^{1}, \ldots, v^{\lastt}, \mathcal{R}^{1}, \ldots, \mathcal{R}^{\lastt}, \Upsilon^{\thist}},
	\end{align*}
	which cancels the expected change in the algorithm's cost. Since neither the potential $\Phi$ nor the cost paid by the algorithm change during rounds in which $\Upsilon^\thist$ does not hold, for all $\thist$ we have the inequality
	\begin{equation}
		\expectover{\substack{v^{\thist}, \mathcal{R}^{\thist}}}*{\Phi(\thist) - \Phi(\lastt) + c(\alg(\thist)) - c(\alg(\lastt)) | v^{1}, \ldots, v^{\lastt}, \mathcal{R}^{1}, \ldots, \mathcal{R}^{\lastt}} \leq 0. \label{eq:loc_perstepbound}
	\end{equation}

	Summing \eqref{eq:loc_perstepbound} for all $t \in [n]$ yields
	\begin{align*}
		\expectover{\substack{v, \mathcal{R}}}*{\Phi(n) - \Phi(0) + c(\alg(n)) - c(\alg(0))} & \leq 0\\
		\expectover{\substack{v, \mathcal{R}}}*{c(\alg(n))} & \leq \Phi(0) + c(\alg(0)) - \expectover{\substack{v, \mathcal{R}}}*{\Phi(n)}.
		\intertext{Combining the above observation that $c(\alg(0)) \leq \beta$ with \Cref{fact:gencost_phi0} then yields}
		\expectover{\substack{v, \mathcal{R}}}*{c(\alg(n))} & \leq O(\beta \cdot \log (mn)).\qedhere
	\end{align*}
\end{proof}

		\section{Covering Integer Programs}
	\label{sec:cips}

	We show how to generalize our algorithm from \cref{sec:gen_cost} to solve pure covering IPs when the constraints are revealed in random order, which significantly generalizes \roosc. 
	Formally, the random order covering IP problem (\roocip) is to solve
	\begin{align}
		\begin{array}{lll}
			\min_z &\langle c, z \rangle \label{eq:CIP}\\
			\text{s.t.} &A z \geq 1  \\
			&z \in \Z_{\geq 0}^m,
			\end{array}
	\end{align}
	when the rows of $A$ are revealed in random order. Furthermore the solution $z$ can only be incremented monotonically and must always be feasible for the subset of constraints revealed so far. (Note that
        we do not consider box constraints, namely upper-bound constraints of the
        form $z_j \leq d_j$.) 
	We may assume without loss of generality that the entries of $A$ are $a_{ij} \in [0,1]$.

We describe an algorithm which guarantees that every row is covered to extent $1-\dthrsh$, meaning it outputs a solution $z$ with $Az \geq 1-\dthrsh$ (this relaxation is convenient in the proof for technical reasons). With foresight, we set $\dthrsh = (e-1)^{-1}$. It is straightforward to wrap this algorithm in one that buys $\lceil(1-\dthrsh)^{-1}\rceil=3$ copies of every column and truly satisfies the constraints, which only incurs an additional factor of $3$ in the cost.

Once again, by a guess-and-double approach, we assume we know a bound
$\beta$ such that $\lpopt \leq \beta \leq 2 \cdot \lpopt$; here $\lpopt$ is the cost of the optimal solution to the LP relaxation of \eqref{eq:CIP}. Let $z^{\thist}$ be the integer solution held by the algorithm at the end of round $t$. 
Define $\Delta_i^{\thist} := \max(0, 1 - \langle a_i, z^{\thist}\rangle)$ to be the extent to which $i$ remains uncovered at the end of round $t$. This time we redefine $\kappa_i^{\thist} := \Delta_i^{\lastt} \cdot \min_k c_k / a_{ik}$, which becomes the minimum
fractional cost of covering the current deficit for
$i$. 
Finally, for a vector $y$, denote the fractional remainder by $\widetilde y := y - \lfloor y \rfloor$. 

The algorithm once again maintains a fractional vector $x$ which is a guess for the (potentially infeasible) LP solution of
cost $\beta$ to \eqref{eq:CIP}. This time, when the $i^{th}$ row arrives at time $t$ and $\Delta_i^{\lastt} > \dthrsh$ (meaning this row is \emph{not} already covered to extent $1-\dthrsh$), we (a) buy a random number of copies of every
column $j$ with probability proportional to its LP value
$x_j$, (b) increase the value $x_j$
multiplicatively and renormalize, and finally (c) buy a minimum cost cover for row $i$ if necessary.
	
	\begin{algorithm}[H]
	\caption{\textsc{LearnOrCoverCIP}}
	\label{alg:cips}
	\begin{algorithmic}[1]
  		\State Initialize $z^0_j \leftarrow \left\lceil \frac{\beta}{c_j \cdot m}\right\rceil \cdot \mathbbm{1} \left\{c_j \leq \frac{\beta}{m}\right\}$. \label{line:cip_cheapsets}
		\State Initialize $x_j^{0} \leftarrow \frac{\beta}{c_j \cdot m'} \cdot \mathbbm{1}\{ c_j \leq \beta\}$, where $m' = |\{j: c_j \leq \beta\}|$.
		\For{$t=1,2\ldots, n$} 
		\State{$i \leftarrow$ $\thist^{th}$ constraint in the
                  random order.}
        \If {$\Delta_i^{\lastt} > \dthrsh$}  \label{line:cip_alreadycov}
			\State Let $y := \kappa^{\thist}_i \cdot x^{\lastt} / \beta $. \ForEach column $j$, add $z^{\thist}_j \leftarrow z^{\lastt}_j + \lfloor y_j\rfloor  + \Ber(\tilde y_j)$. \label{line:cip_updatez}
			\If {$ \langle a_i, x^{\lastt}\rangle< \Delta_i^{\lastt}$} \label{line:cip_noxupdate}
			\State For every $j$, update $x^{\thist}_j \leftarrow x^{\lastt}_j \cdot \exp\left\{\kappa_i^{\thist} \cdot \frac{a_{ij}}{c_j }\right\}$.
			\State Let $Z^{\thist} = \langle c, x^{\thist} \rangle / \beta$ and normalize $x^{\thist} \leftarrow x^{\thist} / Z^{\thist}$. \label{line:normalize}
			\Else 
			\State $x^{\thist} \leftarrow x^{\lastt}$
			\EndIf
			\State Let $k^* = \argmin_k \frac{c_k}{a_{ik}}$. Add $z^{\thist}_{k^*} \leftarrow z^{\thist}_{k^*} + \left\lceil \frac{\Delta_i^{\lastt}}{a_{ik^*}}\right\rceil$. \label{line:cip_backup}
		\EndIf
		\EndFor 
	\end{algorithmic}
\end{algorithm}

Note that once again, \cref{line:normalize} ensures \cref{inv:loc_cost} holds. The main theorem of this section is:

\begin{theorem}
	\label{thm:cip}
	\nameref{alg:cips} is a polynomial-time randomized algorithm for \roocip which outputs a solution $z$ with expected cost $O(\beta  \cdot \log (mn))$ such that $3z$ is feasible.
\end{theorem}
	\cref{thm:intro_cip} follows as a corollary, since given any intermediate solution $z^{\thist}$, we can buy the scaled solution $3z^{\thist}$.

	We generalize the proof of \Cref{thm:weighted_cost_poly}. Redefine $x^*$ to be the optimal LP solution to the final, unknown instance \eqref{eq:CIP}, and $U^{\thist} := \{i \mid \Delta_i^{\thist} > \dthrsh\}$ be the elements which are not covered to extent $1-\dthrsh$ at the end of round $t$. With these new versions of $U^{\thist}$ and $\kappa^{\thist}$, the definitions of both $\rho^{\thist}$ and the potential $\Phi$ remain the same as in \eqref{eq:gencost_rhodef} and \eqref{eq:gencost_phi}, except with different constants $C_1$ and $C_2$.
    \begin{align}
			\boxed{
			\Phi(\thist) := C_1 \cdot \wKL{x^*}{x^{\thist}} + C_2 \cdot \beta \cdot  \log \left(\frac{\rho^{\thist}}{\beta} + \frac{1}{m}\right)}
    \end{align}
	
	Once again, we start with a bound on the initial potential.
	\begin{lemma}[Initial Potential]
		\label{fact:cip_phi0}
	The initial potential is bounded as $\Phi(0) = O(\beta\cdot \log (mn))$, and $\Phi(t) \geq -O(\beta\cdot \log m )$ for all rounds $t$.
	\end{lemma}
	The proof is identical verbatim to that of \cref{fact:gencost_phi0}.
	
	It remains to relate the expected decrease in $\Phi$ to the algorithm's cost in every round.	For convenience, let $X_i^{\thist} := \langle a_i, x^{\thist}\rangle$ be the
	amount that $x^{\thist}$ fractionally covers $i$. Define $\Upsilon^{\thist}$ to be the event that for constraint $i^{\thist}$ arriving in round $t$ we have $\Delta^{\lastt}_i > \dthrsh$. The check at \cref{line:cip_alreadycov} ensures that if $\Upsilon^{\thist}$ does not hold, then neither the cost paid by the algorithm nor the potential will change. We focus on the case that $\Upsilon^{\thist}$ occurs, and once again start with the KL divergence.
	
\begin{lemma}[Change in KL]
	\label{lem:ip_klchange}
	For rounds in which $\Upsilon^{\thist}$ holds, the expected change in weighted KL divergence is
	\[
		\expectover{\substack{i^{\thist}, \mathcal{R}^{\thist}}}*{ \wKL{x^*}{x^{\thist}} - \wKL{x^*}{x^{\lastt}} \mid x^{\lastt}, U^{\lastt}, \Upsilon^{\thist}} 
		\leq \expectover{i \sim U^{\lastt}}{(e-1) \cdot \kappa_i^{\thist} \min(X^{\lastt}_i,\Delta_{i}^{\lastt}) - \kappa_i^{\thist}}.
	\]
\end{lemma}
\begin{proof}
	We break the proof into cases. By the check on \cref{line:cip_noxupdate}, if $X_{i^{\thist}}^{\lastt} \geq \Delta_{i^{\thist}}^{\lastt}$, then the vector $x^{\thist}$ is not updated in round $t$, so the change in KL divergence is 0. This means that 
	\begin{align} 
		&\expectover{\substack{i^{\thist}, \mathcal{R}^{\thist}}}*{ \wKL{x^*}{x^{\thist}} - \wKL{x^*}{x^{\lastt}} \mid x^{\lastt}, U^{\lastt}, \Upsilon^{\thist}, X_{i^{\thist}}^{\lastt}  \geq \Delta_{i^{\thist}}^{\lastt}} \notag \\
		&\leq \expectover{i \sim U^{\lastt}}{(e-1) \cdot \kappa_i^{\thist} \min(X^{\lastt}_i,\Delta_i^{\lastt}) - \kappa_i^{\thist} | X_i^{\lastt} \geq  \Delta_i^{\lastt}} \label{eq:wKL_xigewi}
	\end{align}
	holds trivially, since in this case $\min(X_i^{\lastt},\Delta_i^{\lastt}) = \Delta_i^{\lastt} > \dthrsh = (e-1)^{-1}$ and $\kappa_i^{\thist} \ge 0$. Henceforth we focus on the case $X^{\lastt}_i < \Delta_i^{\lastt}$.
	
	The change in relative entropy depends only on the arriving
        uncovered constraint $i^{\thist}$, not on the randomly chosen
        columns $\mathcal{R}^{\thist}$. Expanding definitions,
	\begin{align}
		& \expectover{\substack{i^{\thist}, \mathcal{R}^{\thist}}}*{\wKL{x^*}{x^{\thist}} - \wKL{x^*}{x^{\lastt}} \mid x^{\lastt}, U^{\lastt}, \Upsilon^{\thist}, X^{\lastt}_i < \Delta_i^{\lastt}} \notag \\
		&= \expectover{i \sim U^{\lastt}}*{\sum_{j} c_j \cdot x^*_j \cdot \log \frac{x^{\lastt}_j}{x^{\thist}_j} | X^{\lastt}_i < \Delta_i^{\lastt}} \notag  \\
		&= \expectover{i \sim U^{\lastt}}*{\sum_{j} c_j \cdot x^*_j \cdot  \log Z^{\thist} - \sum_j c_j \cdot x^*_j \cdot \kappa_i^{\thist} \cdot \frac{a_{ij}}{c_j}  | X^{\lastt}_i < \Delta_i^{\lastt}} \notag \\
		&\leq \expectover{i \sim U^{\lastt}}*{
			 \beta \cdot \log Z^{\thist} - \kappa_i^{\thist} \cdot \sum_j a_{ij} x^*_j  | X^{\lastt}_i < \Delta_i^{\lastt}
		} \label{eq:kl_costofopt} \\
	\intertext{where we used that $\langle c, x^* \rangle \leq \beta$. Expanding the definition of $Z^{\thist}$ and applying the fact that $x^*$ is a feasible solution i.e. $\langle a_{i}, x^* \rangle \geq 1$, we continue to bound \eqref{eq:kl_costofopt} as}
		&\leq \expectover{i \sim U^{\lastt}}*{
			\beta \cdot \log \left(  \frac{1}{\beta} \sum_j c_j x^{\lastt}_j \exp\left( \kappa_i^{\thist} \cdot \frac{a_{ij}}{c_j} \right) \right) - \kappa_i^{\thist}  | X^{\lastt}_i < \Delta_i^{\lastt}
		} \notag \\
		&\leq \expectover{i \sim U^{\lastt}}*{
			\beta \cdot \log \left(  1 + \frac{e-1}{\beta} \cdot \kappa_i^{\thist} \cdot \sum_j a_{ij} x^{\lastt}_j \right) - \kappa_i^{\thist}  | X^{\lastt}_i < \Delta_i^{\lastt}
		} \label{eq:kl_expapprox}. \\
		\intertext{\eqref{eq:kl_expapprox} is derived by applying the
			approximation $e^y \leq 1 + (e-1) y$ for $y \in [0,1]$ and
			the fact that $\langle c, x^{\lastt}\rangle = \beta$ by \cref{inv:loc_cost}; the exponent lies in
			$[0,1]$ because by definition $\kappa_i^{\thist} \cdot a_{ij} / c_j = \Delta_i^{\lastt} \cdot (a_{ij} / c_j) \cdot \min_k (c_k / a_{ik}) \leq 1$ since $\Delta_i^{\lastt} \in [0,1]$. Finally, using the fact that $\log(1+y) \leq y$, we have that \eqref{eq:kl_expapprox} is at most}
		&\leq \expectover{i \sim U^{\lastt}}*{
			(e-1) \cdot \kappa_i^{\thist} \cdot \sum_j a_{ij} x^{\lastt}_j  - \kappa_i^{\thist}  | X^{\lastt}_i < \Delta_i^{\lastt}
		} \notag \\
		&\leq \expectover{i \sim U^{\lastt}}*{
			(e-1) \cdot \kappa_i^{\thist} \cdot \min \left(X^{\lastt}_i, \Delta_i^{\lastt}\right) - \kappa_i^{\thist}  | X^{\lastt}_i < \Delta_i^{\lastt}
		} .\label{eq:wkl_xilwi} 
	\end{align}
	
	The lemma statement follows by combining \eqref{eq:wKL_xigewi} and \eqref{eq:wkl_xilwi} using the law of total expectation.
\end{proof}

We move to bounding the expected change in $\log (\rho^{\thist}/\beta + 1/m)$ provided by updating the solution $z$ on \cref{line:cip_updatez} on the arrival of the random row $i$. Recall that $U^{\thist} = \{i \mid \Delta_i^{\thist} > \dthrsh\}$ are the unseen elements which are at
most covered to extent $1-\gamma$ by $z$.

We will make use of the following fact, which we prove in \cref{sec:appendix}:
\begin{restatable}{fact}{crslem}
	\label{lem:ip_expectcov}
	Given probabilities $p_j$ and coefficients $b_j\in [0,1]$, let
	$W := \sum_j b_j \Ber(p_j)$ be the sum of independent weighted
	Bernoulli random variables.
	Let $\Delta\geq \dthrsh = (e-1)^{-1}$ be some constant.
	Then
	\[
	\expect*{\min\left(W, \Delta \right)} \geq \alpha \cdot \min\left(\expect*{W}, \Delta \right),
	\]
	for a fixed constant $\alpha$ independent of the $p_j$ and $b_j$.
\end{restatable}

We are ready to bound the expected change in $\log (\rho^{\thist}/\beta + 1/m)$. 

\begin{lemma}[Change in $\log \rho^{\thist}$]
	\label{lem:ip_rhochange}
	For rounds in which $\Upsilon^{\thist}$ holds, the expected change in $\log \rho^{\thist}$ is
	\[
		\expectover{\substack{i^{\thist},\mathcal{R}^{\thist}}}*{\log \rho^{\thist} - \log \rho^{\lastt} \mid x^{\lastt}, U^{\lastt}, \Upsilon^{\thist}} \leq - \frac{\alpha}{2\beta}\cdot \expectover{i' \sim U^{\lastt}}*{\kappa_{i'}^{\thist} \cdot \min\left(X_{i'}^{\lastt}, \Delta_{i'}^{\lastt}\right)}
	\]
	where $\alpha$ is a fixed constant.
\end{lemma}

\begin{proof}	
	Conditioned on $i^{\thist} = i$, the expected change in $\log \rho^{\thist}$ depends only on $\mathcal{R}^{\thist}$.
	\begin{align}
		& \expectover{i^{\thist}, \mathcal{R}^{\thist}}*{\left(\frac{\rho^{\thist}}{\beta} + \frac{1}{m}\right) - \log \left(\frac{\rho^{\lastt}}{\beta} + \frac{1}{m}\right) \mid x^{\lastt}, U^{\lastt}, \Upsilon^{\thist}, i^{\thist} = i} \notag \\
		&= \expectover{\mathcal{R}^{\thist}}*{\log \left(1-\frac{\rho^{\lastt} - \rho^{\thist}}{\rho^{\lastt} + \frac{\beta}{m}} \right) | U^{\lastt}, i^{\thist} = i} \notag \\
		&\leq - \frac{1}{\rho^{\lastt} + \frac{\beta}{m}} \cdot \expectover{\mathcal{R}^{\thist}}*{\rho^{\lastt} - \rho^{\thist} | U^{\lastt}, i^{\thist} = i} \label{eq:ip_ut_logapx}. \\
		\intertext{Above, follows from the approximation
			$\log (1-y) \leq -y$. Expanding definitions again, and using the fact that $\kappa_{i'}^{\lastt} - \kappa_{i'}^{\thist} =  \min_k(c_k / a_{i'k})\cdot (\Delta_{i'}^{\lastt} - \Delta_{i'}^{\thist}) \geq \kappa_{i'}^{\thist} (\Delta_{i'}^{\lastt} - \Delta_{i'}^{\thist})$, we further bound \eqref{eq:ip_ut_logapx} by}
		&\leq -  \frac{1}{\rho^{\lastt} + \frac{\beta}{m}} \cdot \expectover{\mathcal{R}^{\thist}}*{\sum_{i' \in U^{\lastt}} \kappa_{i'}^{\thist} \cdot (\Delta_{i'}^{\lastt} - \Delta_{i'}^{\thist})}  \notag\\
		&= -  \frac{1}{\rho^{\lastt} + \frac{\beta}{m}}  \cdot \sum_{i' \in U^{\lastt}} \kappa_{i'}^{\thist} \cdot \expectover{\mathcal{R}^{\thist}}{\Delta_{i'}^{\lastt} - \Delta_{i'}^{\thist}} \notag \\
		&= -  \frac{1}{\rho^{\lastt} + \frac{\beta}{m}} \cdot  \sum_{i' \in U^{\lastt}} \kappa_{i'}^{\thist} \cdot \alpha \cdot  \min\left(\frac{\kappa^{\thist}_{i}}{\beta} \cdot X_{i'}^{\lastt}, \: \Delta_{i'}^{\lastt} \right). \label{eq:ip_randcovg} \\
		\intertext{To understand this last step \eqref{eq:ip_randcovg}, note that $\Delta_{i'}^{\lastt} - \Delta_{i'}^{\thist} = \min(\sum_j a_{i'j} \lfloor y_j \rfloor + \sum_j a_{i'j} \Ber(\widetilde y) ,\: \Delta_{i'}^{\lastt} )$. 
		By the definition of $y$, the first term inside the minimum has expectation $\frac{\kappa_{i}^{\thist}}{\beta} \cdot X^{\lastt}_{i'}$, and since $\Upsilon^{\thist}$ holds we have $\Delta_{i'}^{\lastt} > \dthrsh$.
		Therefore applying \Cref{lem:ip_expectcov} gives \eqref{eq:ip_randcovg} (where $\alpha$ is the constant given by the lemma). 
		Since $\kappa_i^{\thist} / \beta \leq 1$, we bound \eqref{eq:ip_randcovg} with}
		&\leq - \frac{\alpha}{\beta} \cdot \kappa_i^{\thist} \cdot  \frac{1}{\rho^{\lastt} + \frac{\beta}{m}} \cdot \sum_{i' \in U^{\lastt}} \kappa_{i'}^{\thist} \cdot  \min\left(X_{i'}^{\lastt}, \Delta_{i'}^{\lastt}\right) \notag \\ 
		&= - \frac{\alpha}{\beta} \cdot \kappa_i^{\thist} \cdot \frac{|U^{\lastt}|}{\rho^{\lastt} + \frac{\beta}{m}} \cdot \expectover{i' \sim U^{\lastt}}*{\kappa_{i'}^{\thist} \cdot \min \left(X_{i'}^{\lastt}, \Delta_{i'}^{\lastt})\right)}. \label{eq:ip_kcrs}
	\end{align}
	
	Taking the expectation of \eqref{eq:ip_kcrs} over $i \sim U^{\lastt}$, and using the fact that $\expectover{i \sim U^{\lastt}}*{\kappa_i^{\thist}} = \rho^{\lastt} / U^{\lastt}$, the expected change in $\log \rho^{\thist}$ becomes
	\begin{align*}
		&\expectover{i^{\thist}, \mathcal{R}^{\thist}}*{\log \rho^{\thist} - \log \rho^{\lastt} \mid x^{\lastt}, U^{\lastt}, \Upsilon^{\thist}} 
		\leq -  \frac{\alpha}{\beta} \cdot\frac{\rho^{\lastt}}{\rho^{\lastt} + \frac{\beta}{m}} \cdot \expectover{i' \sim U^{\lastt}}*{\kappa_{i'}^{\thist} \cdot \min\left(X_{i'}^{\lastt}, \Delta_{i'}^{\lastt}\right)} \\
		&\leq -  \frac{\alpha}{2\beta} \cdot \expectover{i' \sim U^{\lastt}}*{\kappa_{i'}^{\thist} \cdot \min\left(X_{i'}^{\lastt}, \Delta_{i'}^{\lastt}\right)},
	\end{align*}    
	as desired. In the last step, we used the fact that $\Upsilon^{\thist}$ holds, so $U^{\lastt}$ is nonempty, and in particular $\rho^{\lastt} \geq \frac{\beta}{m}$. To see this, note that $\Delta_i^{\lastt} \geq \dthrsh \geq \nicefrac{1}{2}$ for any $i \in U^{\lastt}$. If \cref{line:cip_cheapsets} did not already cover constraint $i$, then for all $j \in [m]$, we have $\beta \cdot a_{ij} / (c_j \cdot m) \leq \nicefrac{1}{2}$ which means that $c_j / a_{ij} \geq 2 \beta / m$ and thus
    \[
        \rho^{\lastt} \geq \kappa_i^{\lastt} = \Delta_i^{\lastt} \cdot \min_j \frac{c_j}{a_{ij}} \geq \gamma \cdot \frac{2\beta}{m} \geq \frac{\beta}{m}. \qedhere
    \]
\end{proof}	

We may now combine the two previous lemmas as before.	
\begin{proof}[Proof of \cref{thm:cip}]
We start with the cost of \Cref{line:cip_cheapsets}. For each $j$ which the algorithm initializes to $z_j^0 > 0$, it buys $\left\lceil \frac{\beta}{m \cdot c_j}\right\rceil$ copies, which costs
\[
    c_j \cdot \left\lceil \frac{\beta}{m \cdot c_j}\right\rceil \leq c_j \cdot \left( \frac{\beta}{m \cdot c_j} + 1\right) = \frac{\beta}{m} + c_j \leq \frac{2\beta}{m},
\]
since it only buys $j$ for which $c_j \leq \frac{\beta}{m}$. So overall the total cost of \Cref{line:cip_cheapsets} is at most $2\beta$. 

In the round in which constraint $i$ arrives, the expected cost of sampling is $\frac{\kappa_i^{\thist}}{\beta} \langle c, x^{\lastt}\rangle = \kappa_i^{\thist}$ (by \cref{inv:loc_cost}). The algorithm pays an additional \[ c_{k^*} \left\lceil \frac{\Delta_i^{\lastt}}{a_{ik^*}}\right\rceil = c_{k^*} \left\lceil \frac{\kappa_i^{\thist}}{c_{k^*}}\right\rceil \leq 2 \kappa_i^{\thist}\] in \cref{line:gencost_backup}, where this upper bound holds because $\frac{\kappa_i^{\thist}}{c_{k^*}} = \frac{\Delta_i^{\lastt}}{a_{ik^*}}\geq \nicefrac{1}{2}$, since $\Delta_i^{\lastt} \geq \dthrsh \geq \nicefrac{1}{2}$ and $a_{ik^*} \leq 1$. 
Hence the total expected cost per round is at most $3 \cdot \kappa_i^{\thist}$.

Combining \cref{lem:ip_klchange} and \cref{lem:ip_rhochange} and
choosing $C_1 = 3$ and $C_2 = 6(e-1) / \alpha$, we have
\begin{align*}
	&\expectover{\substack{i^{\thist}, \mathcal{R}^{\thist}}}*{\Phi(t) - \Phi(\lastt) | i^{1}, \ldots, i^{\lastt}, \mathcal{R}^{1}, \ldots, \mathcal{R}^{\lastt}, \Upsilon^{\thist}} \\
	& = \expectover{\substack{i^{\thist}, \mathcal{R}^{\thist}}}*{
		\begin{array}{ll}
			&C_1 \cdot \left(\wKL{x^*}{x^{\thist}} - \wKL{x^*}{x^{\lastt}}\right) \\
			+ & C_2 \cdot \beta \cdot \left(\log \left(\frac{\rho^{\thist}}{\beta} + \frac{1}{m} \right) - \log \left(\frac{\rho^{\lastt}}{\beta} + \frac{1}{m}\right)\right)
		\end{array}
		| i^{1}, \ldots, i^{\lastt}, \mathcal{R}^{1}, \ldots, \mathcal{R}^{\lastt}, \Upsilon^{\thist}} \\
	&\leq - \expectover{\substack{i^{\thist}, \mathcal{R}^{\thist}}}*{3 \cdot \kappa_{i^{\thist}}^{\thist} | i^{1}, \ldots, i^{\lastt}, \mathcal{R}^{1}, \ldots, \mathcal{R}^{\lastt}, \Upsilon^{\thist}},
\end{align*}
which cancels the expected change in the algorithm's cost. Since neither the potential $\Phi$ nor the cost paid by the algorithm change during rounds in which $\Upsilon^\thist$ does not hold, for all $\thist$ we have the inequality
\[
	\expectover{\substack{i^{\thist}, \mathcal{R}^{\thist}}}*{\Phi(t) - \Phi(\lastt) + c(\alg(t)) - c(\alg(\lastt)) | i^{1}, \ldots, i^{\lastt}, \mathcal{R}^{1}, \ldots, \mathcal{R}^{\lastt}} \leq 0.
\]
Summing over all $t \in [n]$ gives
\begin{align*}
	\expectover{\substack{v, \mathcal{R}}}*{\Phi(n) - \Phi(0) + c(\alg(n)) - c(\alg(0))} & \leq 0\\
	\expectover{\substack{v, \mathcal{R}}}*{c(\alg(n))} & \leq \Phi(0) + c(\alg(0)) - \expectover{\substack{v, \mathcal{R}}}*{\Phi(n)}.
	\intertext{Combining the observation that $c(\alg(0)) \leq 2\beta$ (due to \cref{line:cip_cheapsets}) with \Cref{fact:cip_phi0} then yields}
	\expectover{\substack{v, \mathcal{R}}}*{c(\alg(n))} & \leq O(\beta \cdot \log (mn)). \qedhere
\end{align*}
\end{proof}

\section{Lower Bounds}
\label{sec:lbs}

We turn to showing lower bounds for \setcov and related problems in the random order model.
The lower bounds for \roosc are proven via basic probabilistic and combinatorial arguments.
We also show hardness for a batched version of \roosc, which has implications for related problems.

\subsection{Lower Bounds for \roosc}

\label{sec:ro_lb}

We start with information theoretic lower bounds for the RO setting.

\begin{theorem} \label{thm:fraclogLB}
	The competitive ratio of any randomized fractional or integral algorithm for \roosc is $\Omega(\log n)$.
\end{theorem}

\begin{proof}	
	Consider the following instance of \roosc in which $m = 2^\ell$ and $n =2^\ell - 1$. Construct the instance recursively in $\ell$ rounds. Define the sub collection of sets $\mathcal{S}_0 = \mathcal{S}$, i.e. initially all the sets. For each round $i$ from $1$ to $\ell$, do: (a) create $2^{\ell-i}$ new elements and add each to every set in $\mathcal{S}_{i-1}$, and (b) choose $\mathcal{S}_{i}$ to be a uniformly random subcollection of $\mathcal{S}_{i-1}$ of size $|\mathcal{S}_{i-1}|/2$.
	
	By Yao's principle, it suffices to lower bound the cost of any fixed deterministic algorithm $A$ that maintains a monotone LP solution in random order. Let $x^t$ be the fractional solution of $A$ at the end of round $t$. Assume $A$ is lazy in the sense that in every round and for every coordinate $x_S$ that is incremented in that round, setting that coordinate to $x_S - \epsilon$ is infeasible for all $\epsilon> 0$ with respect to the elements observed up until and including this round. 
	
	We refer to the elements added in round $i$ as the type $i$ elements. 
	Let $\mathfrak{g}(i)$ be the event that at least one element of that type arrives in random order before any element of type $j$ for $j > i$. 
	Note that $\prob{\mathfrak{g}(i)} = 2^{l-i} / (2^{l-i+1} - 1)> \nicefrac{1}{2}$ for all $i$. 
	
	Also define $c(A,i)$ be the cost paid by the algorithm to cover the first element of type $i$ that arrives in random order (potentially $0$). 
	Conditioned on $\mathfrak{g}(i)$ holding, we claim the expected cost $c(A,i)$ is at least $\nicefrac{1}{2}$. 
	To see this, first suppose that $i$ is the first type for which $\mathfrak{g}(i)$ holds; then this is the first element to arrive, and so $\sum_{S \in \mathcal{S}_0} x^t_S = 0$ beforehand and $\sum_{S \in \mathcal{S}_i} x^t_S = 1$ afterward, and so $c(A,i) = 1 \geq 1/2$.
	
	Otherwise, let $k < i$ be the last type for which $\mathfrak{g}(k)$ held. 
	Since $A$ is lazy, at the time $t$ before the first element of type $i$ arrived, we had $\sum_{S \in \mathcal{S}_k} x^t_S = 1$. 
	By the construction of the instance and the fact that $i >k$, the collection $\mathcal{S}_i$ consists of a uniformly random subset of $\mathcal{S}_k$ of size at most $|\mathcal{S}_k|/2$. 
	Deferring the random choice of this collection $\mathcal{S}_i$ until this point, we have that $\expect{\sum_{S \in \mathcal{S}_i} x^t_S} \leq \nicefrac{1}{2}$. 
	Hence $\displaystyle \expect*{c(A,i) | \mathfrak{g}(i)} \geq \nicefrac{1}{2}$.

	To conclude, the optimal solution consists of the vector that has $1$ on the one coordinate $S$ containing all the elements and $0$ elsewhere, whereas
	\[
		\expect{c(A)} \geq \sum_{i} \prob{\mathfrak{g}(i)} \cdot \expect*{c(A,i) | \mathfrak{g}(i)} > \frac{\ell}{4} = \Omega(\log n). 
	\]
	Since the optimum is integral, the competitive ratio lower bound also holds for integral algorithms.
\end{proof}

We emphasize that this set system has a VC dimension of $2$, which rules out improved algorithms for set systems of small VC dimension in this setting.

\begin{theorem}
    The competitive ratio of any randomized algorithm for \roosc is $\Omega\left(\frac{\log  m}{\log \log m} \right)$ even when $m \gg n$.
\end{theorem}

The following proof uses the construction in Proposition 4.2 of \cite{alon2003online}; they take the product of this construction with another to show a stronger bound for the adversarial order setting. It is a simple observation that the part of the construction we use  gives a bound even in random order, but we include the proof for completeness.

\begin{proof}
	Given integer parameter $r$, let $\mathcal{S} = \binom{[10r^2]}{r}$ be the set of all subsets of $[10r^2]$ of size $r$. The adversary chooses $U$ to be a random subset of $[10r^2]$ of size $r$ and reveals it in random order.
	
	By Yao's principle, it again suffices to bound the performance of any deterministic algorithm $A$, and we may again assume that $A$ is lazy. By deferring randomness, the adversary is equivalent to one which randomly selects an element $r$ times without replacement from $[10r^2]$. Since the algorithm selects at most $r$ sets each of size at most $r$, the number of elements of $[10r^2]$ covered by $A$ is at most $[r^2]$, and hence every element chosen by the adversary has probability at least $4/5$ of being uncovered on arrival. Hence the algorithm selects at least $4r/5$ sets in expectation, whereas $\opt$ consists of the single set covering the $r$ elements of the adversary, so the competitive ratio of $A$ is $\Omega(r)$. The claim follows by noting that $r = \Omega(\log m / \log \log m)$.
\end{proof}

\subsection{Performance of \cite{buchbinder2009online} in Random Order}

\label{sec:bn_lb}

In this section we argue that the algorithm of \cite{buchbinder2009online} has a performance of $\Omega(\log m \log n)$ for \roosc in general. One instance demonstrating this bound is the so called upper triangular instance $\Delta = (U_{\Delta}, \mathcal{S}_\Delta)$ for which $n=m$ and which we now define. Let the sets $\mathcal{S}_{\Delta} = \{S_1, \ldots, S_n \}$ be fixed. 
Choose a random permutation $\pi \in S_n$. 
Then for every $i=1, \ldots, m$, let $S_{\pi(i)} = \{ n-i+1, \ldots, n\}$, since it will be convenient for elements to appear in \cref{fig:permutedtriangular} in descending order.

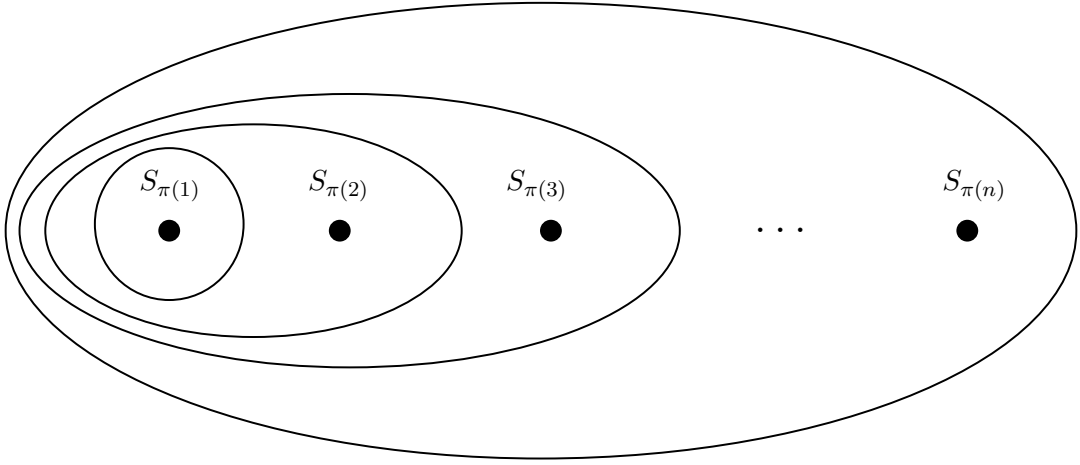
\begin{figure}[H]
	\centering
	\tikzset{every picture/.style={line width=0.75pt}} 

\begin{tikzpicture}[x=0.75pt,y=0.75pt,yscale=-1,xscale=1]
	
	\draw   (75,151.67) .. controls (75,130.5) and (91.79,113.33) .. (112.5,113.33) .. controls (133.21,113.33) and (150,130.5) .. (150,151.67) .. controls (150,172.84) and (133.21,190) .. (112.5,190) .. controls (91.79,190) and (75,172.84) .. (75,151.67) -- cycle ;
	\draw   (50,155) .. controls (50,125.36) and (97.01,101.33) .. (155,101.33) .. controls (212.99,101.33) and (260,125.36) .. (260,155) .. controls (260,184.64) and (212.99,208.67) .. (155,208.67) .. controls (97.01,208.67) and (50,184.64) .. (50,155) -- cycle ;
	\draw   (37,155) .. controls (37,116.89) and (111.54,86) .. (203.5,86) .. controls (295.46,86) and (370,116.89) .. (370,155) .. controls (370,193.11) and (295.46,224) .. (203.5,224) .. controls (111.54,224) and (37,193.11) .. (37,155) -- cycle ;
	\draw   (30,155) .. controls (30,91.49) and (150.88,40) .. (300,40) .. controls (449.12,40) and (570,91.49) .. (570,155) .. controls (570,218.51) and (449.12,270) .. (300,270) .. controls (150.88,270) and (30,218.51) .. (30,155) -- cycle ;
	\draw  [fill={rgb, 255:red, 0; green, 0; blue, 0 }  ,fill opacity=1 ] (107.5,155) .. controls (107.5,152.24) and (109.74,150) .. (112.5,150) .. controls (115.26,150) and (117.5,152.24) .. (117.5,155) .. controls (117.5,157.76) and (115.26,160) .. (112.5,160) .. controls (109.74,160) and (107.5,157.76) .. (107.5,155) -- cycle ;
	\draw  [fill={rgb, 255:red, 0; green, 0; blue, 0 }  ,fill opacity=1 ] (193.5,155) .. controls (193.5,152.24) and (195.74,150) .. (198.5,150) .. controls (201.26,150) and (203.5,152.24) .. (203.5,155) .. controls (203.5,157.76) and (201.26,160) .. (198.5,160) .. controls (195.74,160) and (193.5,157.76) .. (193.5,155) -- cycle ;
	\draw  [fill={rgb, 255:red, 0; green, 0; blue, 0 }  ,fill opacity=1 ] (300,155) .. controls (300,152.24) and (302.24,150) .. (305,150) .. controls (307.76,150) and (310,152.24) .. (310,155) .. controls (310,157.76) and (307.76,160) .. (305,160) .. controls (302.24,160) and (300,157.76) .. (300,155) -- cycle ;
	\draw  [fill={rgb, 255:red, 0; green, 0; blue, 0 }  ,fill opacity=1 ] (510,155) .. controls (510,152.24) and (512.24,150) .. (515,150) .. controls (517.76,150) and (520,152.24) .. (520,155) .. controls (520,157.76) and (517.76,160) .. (515,160) .. controls (512.24,160) and (510,157.76) .. (510,155) -- cycle ;
	
	\draw (406,152) node [anchor=north west][inner sep=0.75pt]  [font=\LARGE]  {$\dotsc $};
	\draw (96,123) node [anchor=north west][inner sep=0.75pt]    {$S_{\pi(1)}$};
	\draw (281,123) node [anchor=north west][inner sep=0.75pt]    {$S_{\pi(3)}$};
	\draw (501,123) node [anchor=north west][inner sep=0.75pt]    {$S_{\pi(n)}$};
	\draw (181,123) node [anchor=north west][inner sep=0.75pt]    {$S_{\pi(2)}$};

\end{tikzpicture}
	\caption{Tight instance for \cite{buchbinder2009online} in random order.  \label{fig:permutedtriangular}}
\end{figure}

\begin{claim}
	\cite{buchbinder2009online} is $\Omega(\log m \log n)$-competitive on the instance $\Delta = (U_\Delta, \mathcal{S}_\Delta)$ in RO.
\end{claim}

\begin{proof}
	
	The final solution $\mathcal{T}$ output by \cite{buchbinder2009online} on this instance is equivalent to that of the following algorithm which waits until the end of the sequence to buy all of its sets.
	Maintain a (monotone) LP solution $x$ whose coordinates are indexed by sets. 
	Every time an uncovered element $i$ arrives, for all $j\geq i$ increase the weights of all sets $x_{\pi(j)}$ uniformly until $\sum_{j \geq i} x_{\pi(j)} = 1$. 
	Finally, at the end of the element sequence, sample each set with probability $\min(\log n \cdot x_S, 1)$. 
	We claim that this produces a solution of cost $\Omega(\log^2 n)$ in expectation on the instance above when elements are presented in random order, whereas the optimal solution consists of only the one set $S_1$.
	
	Let $X_i$ be the event that $i$ arrives before any element $j$ with $j<i$. 
	This corresponds to $i$ appearing in the fewest sets of any element seen thus far; we call such an element \textit{leading}. 
	Let $k(i) = \min\{k \mid X_k, \ k>i \}$ be the most recent leading index before $i$, if one exists.
	Note that if $X_i$ occurs and $k(i) = k$, then the expected increase in the size of the final solution due to $X_i$ is exactly $ i \cdot \left(\min\left(\frac{\log n}{i}, 1 \right)-\min\left(\frac{\log n}{k}, 1 \right)\right)$. 
	(If $X_i$ occurs but $k(i)$ does not exist, then $i$ is the first leading element and so the expected cost increase in the final solution is $i \cdot \min\left(\frac{\log n}{i}, 1 \right)$.)

	What is $\prob{X_i, \ k(i) = k}$ for some $i < k$? 
	It is precisely the probability that the random arrival order induces an order on only the elements $k, k-1, \ldots, i, \ldots, 1$ which puts $k$ first and $i$ second, which is $1/(k(k-1))$. 
	
	Thus the total expected size of the final solution $\mathcal{T}$ can be bounded by
	\begin{align*}
		\expect{\lvert \mathcal{T} \rvert} &= \sum_{i=1}^{n-1}  \sum_{k = i + 1}^n \prob{X_i, \ k(i) = k} \cdot i \left(\min\left(\frac{\log n}{i}, 1 \right)-\min\left(\frac{\log n}{k}, 1 \right)\right) \\
		&\geq \sum_{i=1}^{n-1}  \sum_{k = i+1}^n \frac{1}{k(k-1)} \cdot i \left(\min\left(\frac{\log n}{i}, 1 \right)-\min\left(\frac{\log n}{k}, 1 \right)\right) \\
		&\geq \log n \cdot \sum_{i= \log n}^{n-1}  \sum_{k =i+ 1}^n \frac{1}{k(k-1)} \cdot i \left(\frac{1}{i} - \frac{1}{k}\right) \\
		&\geq \log n \cdot \Omega(\log n - \log \log n) \\
		&= \Omega(\log^2 n). \qedhere
	\end{align*}
\end{proof}

\subsection{Lower Bounds for Extensions}

\label{sec:lb_ext}

We study lower bounds for several extensions of \roosc. Our starting point is a lower bound for the batched version of the problem. Here the input is specified by a set system $(U, \mathcal{S})$ as before, along with a partition of $U$ into batches $B_1, B_2, \ldots B_b$. For simplicity, we assume all batches have the same size $s$. The batches are revealed one-by-one in their entirety to the algorithm, in uniform random order. After the arrival of a batch, the algorithm must select sets to buy to cover all the elements of the batch. 

Using this lower bound, we derive as a corollary a lower bound for the random order version of \subcov defined in \cite{gupta2020online}. It is tempting to use the method of \cref{thm:weighted_cost_poly} to improve their competitive ratio of $O(\log n \log (t \cdot f(N) / \fmin))$ in RO (we refer the reader to \cite{gupta2020online} for the definitions of these parameters). We show that removing a log from the bound is not possible in general.

\begin{theorem}
		\label{lem:batched}
		The competitive ratio of any polynomial-time randomized algorithm on batched \roosc with $b$ batches of size $s$ is $\Omega(\log b \log s)$ unless $\NP \subseteq \BPP$.
\end{theorem}

	We follow the proof of \cite[Theorem~2.3.4]{korman2004use}, which demonstrates that there is no randomized $o(\log m \log n)$-competitive polynomial-time algorithm for (adversarial order) \osetcov unless $\NP \subseteq \BPP$.
	We adapt the argument to account for random order.

	Consider the following product of the upper triangular instance from \cref{sec:bn_lb} together with an arbitrary instance of (offline) \setcov $H$.
	In particular, take $\Delta = (U_\Delta, \mathcal{S}_\Delta)$ to be the upper triangular instance on $N$ elements, and let $H = (U_H, \mathcal{S}_H)$ denote a set cover instance, where $U_H = [N']$ and $|\mathcal{S}_H| = M'$. Note that $\Delta$ is a random instance, where the randomness is over the choice of label permutation $\pi$.
	Define product instance $\Delta \times H = (U_{\Delta \times H}, \mathcal{S}_{\Delta \times H})$ by
	\begin{align*}
		U_{\Delta \times H} &= \{(i,j) \in U_\Delta \times U_H \} \\ 
		\mathcal{S}_{\Delta \times H} & = \{S_{ij} = S_i \times S_j: (S_i, S_j) \in \mathcal{S}_\Delta \times \mathcal{S}_H \}.
	\end{align*}
	Observe that $|U_{\Delta \times H}| = N N'$ and $| \mathcal{S}_{\Delta \times H} | = N M'$. 

	Each copy of $H$ is a batch of this instance, so that batches of elements $B_1, \ldots, B_N$ are given by $B_i = \{(i,j): j \in U_H\}$. Thus the parameters of the batched \roosc instance are $b = N$ and $s = N'$.
	These batches $B_i$ will arrive in uniformly random order according to some permutation $\sigma \in S_N$. 
	For a randomized algorithm $A$, let 
	\[
		C(A(\Delta\times H)) := \expectover{\substack{\sigma, \pi \in S_N \\ \mathcal{R}}}{c(A(\Delta \times H))}
	\]
	denote the expected cost of $A$ on the instance $\Delta \times H$, where $\mathcal{R}$ denotes the randomness of $A$.

	We begin by establishing an information-theoretic lower bound, whose proof we defer to \cref{sec:appendix}.
	
	\begin{restatable}{lemma}{productLB}
		\label{lem:productLB}
		Let $A$ be a randomized algorithm for batched \roosc.
		Then on the instance $\Delta \times H$,
		\[
		C(A(\Delta \times H)) \geq \frac{1}{2}  \cdot |\opt(H)| \cdot \log N.
		\]
	\end{restatable}
	
	We also require the following theorem of \cite{raz1997sub}:
	\begin{lemma} \label{lem:sattosetcover}
		There exists a polynomial-time reduction from \SAT to \setcov that, given a formula $\psi$, produces a \setcov instance $H^\psi$ with $N'$ elements and $M'$ sets for which
		\begin{itemize}
			\item $M' = N'^\alpha$ for some constant $\alpha$,
			\item if $\psi \in \SAT$ then $\opt(H^\psi) = K(N')$, 
			\item if $\psi \not \in \SAT$ then $\opt(H^\psi) \geq c \cdot \log N' \cdot  K(N')$, 
		\end{itemize}
		for some polynomial-time-computable $K$ and constant $c\in(0,1)$. 
	\end{lemma}
	
	We are ready to show the lower bound.
	
	\begin{proof}[Proof of \cref{lem:batched}]
		
	Suppose that there is some polynomial-time randomized algorithm $A$ with competitive ratio at most $c/4 \cdot \log b \log s$ for batched \roosc (recall that $c$ is the constant given in \cref{lem:sattosetcover} above). Then, we argue, the following is a $\BPP$ algorithm deciding \SAT.
	
	Given a formula $\psi$, we first reduce it to an instance of batched \roosc: first feed the formula through the reduction in \cref{lem:sattosetcover} to get the instance $H^\psi$,  then create the batched \roosc instance defined by $\Delta \times H^{\psi}$. Finally, run $A$ on $\Delta \times H^{\psi}$ a number $W = \poly(n)$ times, and let $\overline C$ be the empirical average of $c(A(\Delta \times H^{\psi}))$ over these runs. If $\overline C \geq 3c/8 \cdot \log b \log s \cdot K(N')$, output $\psi \in \SAT$, else output $\psi \not \in \SAT$. It suffices to argue that this procedure answers correctly with high probability.
	
	\begin{claim}
		If $\psi \in \SAT$, then $C(A(\Delta, H^{\psi})) \leq \frac{c}{4} \cdot \log b \log s \cdot K(N')$.
	\end{claim}

	\begin{proof}
		By assumption, $A$ has the guarantee
		\begin{align*}
			C(A(\Delta \times H^\psi)) &\leq \frac{c}{4} \log b \log s \cdot |\opt(\Delta \times H^\psi)| \\
			& = \frac{c}{4} \cdot \log b \log s \cdot K(N'). \qedhere
		\end{align*}
	\end{proof}

	\begin{claim}
	If $\psi \not \in \SAT$, then $C(A(\Delta, H^{\psi})) \geq \frac{c}{2} \cdot \log b \log s \cdot K(N')$.
	\end{claim}
	
	\begin{proof}
		By \cref{lem:productLB}, the performance of $A$ is lower bounded as
		\begin{align*}
			C(A(\Delta \times H^\psi)) &\geq \frac{1}{2} H_N \cdot |\opt(H^\psi)| \\
			&\geq \frac{1}{2} \log N \cdot c \cdot \log (N') \cdot K(N') \\
			&\geq \frac{c}{2} \cdot  \log b \log s  \cdot K(N'). \qedhere
		\end{align*}
	\end{proof}
	
	To complete the proof, note that $c(A(\Delta \times H^\psi)) \in [0, NM']$. Setting $W = \poly(n)$ sufficiently high, by a Hoeffding bound, the estimate $\overline C$ concentrates to within $\frac{c}{8} \cdot \log b \log s \cdot K(N')$ of $C(A(\Delta, H^{\psi}))$ with high probability, in which case the procedure above answers correctly.
\end{proof}

We now use \cref{lem:batched} to derive lower bounds for \rosubc.

\begin{corollary}
	The competitive ratio of any polynomial-time randomized algorithm against \rosubc is $\Omega(\log n \cdot \log (f(N)/\fmin))$ unless $\NP \subseteq \BPP$.
\end{corollary}

\begin{proof}
	Batched \roosc is a special case of online \subcov in which $f_i$ is the coverage function of block $i$. In this case the parameter $f(N)/\fmin = s$, so the statement follows by applying \cref{lem:batched} with $b = s = \sqrt{n}$.
\end{proof}

\section{Conclusion}

In this work we introduce \nameref{alg:gencost} as a method for solving \roosc and \roocip with competitive ratio nearly matching the best possible offline bounds. On the other hand we prove nearly tight \emph{information theoretic} lower bounds in the RO setting. We also show lower bounds and separations for several generalizations of \roosc. We leave as an interesting open question whether it is possible to extend the technique to covering IPs \emph{with box constraints}.

We hope our method finds uses elsewhere in online algorithms for RO settings. In \cref{sec:disussion} we discuss suggestive connections between our technique and other methods in online algorithms, namely projection based algorithms and stochastic gradient descent.

\bigskip

\textbf{Acknowledgements:} Roie Levin would like to thank Vidur Joshi for asking what is known about online set cover in random order while on the subway in NYC, as well as David Wajc and Ainesh Bakshi for helpful discussions.
	
	{\footnotesize
		\bibliography{refs}
		\bibliographystyle{alpha}
	}
	
	\appendix

\section{Deferred Proofs}

\label{sec:appendix}

Our potential function arguments make use of the following facts.

\lpsupport*

\begin{proof}
	Suppose otherwise and let $x^*$ be an optimal LP solution. 
	Let $j'$ be a coordinate for which $c_{j'} > \beta$ and $x^*_{j'} > 0$.	
	Then define the vector $x'$ by
	\[x'_j = \begin{cases}
		0 &  \text{if $j = j'$} \\
		\frac{x_j^*}{1-x_j'^*} & \text{otherwise},
	\end{cases}\]
	First of all, note that $x^*_{j'} < 1$ since $x^*_{j'} c_{j'} \leq \sum_j c_j x^*_j \leq \beta < c_{j'}$, and so $x' \geq 0$. To see that $x'$ is feasible for each constraint $\langle a_i, x\rangle \geq 1$,
	\begin{align*}
		\langle a_i, x^*\rangle = (1-x^*_{j'}) \langle a_i, x'\rangle + a_{ij'} x^*_{j'} \geq 1 \quad \text{so} \quad
		\langle a_i, x'\rangle \geq \frac{1}{1 - x^*_{j'}} (1 - a_{ij'} x^*_{j'}) \geq 1,
	\end{align*}
	since $a_{ij'} \in [0,1]$. Finally, observe that $x'$ costs strictly less than $x^*$, since 
	\begin{align*}
		\langle c, x' \rangle  = \frac{\langle c, x^* \rangle - x^*_{j'} c_{j'} }{1-x^*_{j'}} < \frac{\langle c, x^* \rangle - x^*_{j'} \langle c, x^* \rangle}{1-x^*_{j'}}  = \langle c, x^* \rangle.
	\end{align*}
	
	This contradicts the optimality of $x^*$, and so the claim holds.
\end{proof}

\crslem*

\begin{proof}
	We first consider the case when $\expect*{W} \geq
	\nicefrac{\Delta}{3}$. 
	By the Paley-Zygmund 
	inequality, noting that $\expect*{W} = \sum_j b_j p_j$ and so $\sigma^2 = \sum_j p_j (1-p_j) b_j^2 \leq \expect*{W}$, we have
	\[
		\prob{W \geq \nicefrac{\Delta}{6}} \geq \prob{W \geq \expect*{W}/2} \geq \frac{1}{4} \cdot \frac{\expect*{W}^2}{\expect*{W}^2 + \sigma^2} \geq \frac{1}{4} \cdot \frac{\expect*{W}}{1 + \expect*{W}} \geq \frac{1}{4} \cdot \frac{\nicefrac{\Delta}{3}}{1 + \nicefrac{\Delta}{3}}  \geq \nicefrac{1}{28},
	\]
	since $\Delta \geq \gamma \geq \nicefrac{1}{2}$ by assumption.
	This implies the claim because in this case
	\[
		\expect*{\min\left(W, \Delta \right)} \geq \frac{\Delta}{6} \cdot \prob{W \geq \nicefrac{\Delta}{6}} \geq \frac{\Delta}{168} \geq \frac{1}{168} \cdot \min(\expect*{W}, \Delta).	
	\]
	
	Otherwise $\expect*{W} < \nicefrac{\Delta}{3}$. 
	Let $\mathcal{R}$ denote the random subset of $j$ for which
	$\Ber(p_j) = 1$ in a given realization of $W$, and let
	$\mathcal{K}$ be the random subset of the $j$ which is output
	by the $(\nicefrac{1}{3},\nicefrac{1}{3})$-contention resolution scheme for knapsack constraints when given $\mathcal{R}$ as input, as defined in  \cite[Lemma 4.15]{chekuri2014submodular}. 
	The set $\mathcal{K}$ has the properties that (1) (over the randomness in $\mathcal{R}$) every $j$ appears in $\mathcal{K}$ with probability at least $p_j/3$, (2) $\sum_{j \in \mathcal{K}} a_{ij} < \Delta$, and (3) $\mathcal{K} \subseteq \mathcal{R}$.
	Hence in this case
	\[
		\expect*{\min\left(W, \Delta \right)} \geq \expectover{\mathcal{K}}*{\sum_{j \in \mathcal{K}} a_{ij}} \geq \frac{ \expect*{W} }{3}= \frac{1}{3} \cdot \min(\expect*{W}, \Delta). \qedhere
	\]
	
	Therefore the claim holds with $\alpha = \nicefrac{1}{168}$. 
\end{proof}

\productLB*

\begin{proof}
	By Yao's principle, it suffices to bound the expected performance of any deterministic algorithm $A$ over the randomness of the instance.
	Here the performance is in expectation over the random choice of instance as well as the random order of batch arrival. 
	
	The randomness in the input distribution $\Delta \times H$ is over the set labels, given by the random permutation $\pi \in S_N$. For convenience, we instead equivalently imagine the set labels in $\Delta$ are fixed, and that $\pi$ is a random permutation over batch labels, so that the label of batch $i$ is $l_i = \pi(i)$. This means that for a fixed realization of $\pi$, for any two sets $S_{lj}, S_{l'j'} \in \mathcal{S}_{\Delta \times H}$, even before any batches have arrived $A$ can determine whether $l = l'$, but because $\pi \sim S_N$ uniformly at random, $A$ cannot determine which batches $i$ the sets $S_{lj}$ and $S_{l'j'}$ intersect.		
	Without loss of generality, we assume that $A$ is lazy, in that for each batch it only buys sets which provide marginal coverage. 
	
	The (offline) instance $H$ has some optimal cover $\{S^*_1, \ldots, S^*_k\} := \opt(H)$.
	Let $\mathcal{S}^* := \{S_{lj} = S_l \times S_j : S_l \in \mathcal{S}_{\Delta}, S_j \in \opt(H) \}$ denote the sets in $\mathcal{S}_{\Delta \times H}$ which project down to sets in $\opt(H)$. 
	We argue that it suffices to consider a deterministic lazy algorithm $A^*$ which only buy sets in $\mathcal{S}^*$. 
	Let $c(A(\pi, \sigma))$ denote the cost of $A$ on batch label permutation $\pi$, when the batch arrival order is $\sigma \in S_N$. 
	For every feasible $A$ we argue that there is some $A^*$ which buys only sets in $\mathcal{S}^*$, is feasible for every batch $\sigma(i)$ upon arrival, and buys at most as many sets as $A$ for any $\pi \in \sigma_N$, so that
	\begin{equation} \label{eqn:buyingfromoptisbetter}
		c(A^*(\pi, \sigma)) \leq c(A^*(\pi, \sigma)).
	\end{equation}
	This will in turn imply that
	\begin{equation}
		\expectover{\substack{\sigma, \pi \sim S_N}}{c(A^*(\pi,\sigma))} \leq \expectover{\substack{\sigma, \pi \sim S_N}}{c(A(\pi, \sigma))},
	\end{equation}
	and so it will suffice to lower bound the performance of any $A^*$ in this restricted class. 
	
	Given $A$, we will construct an $A^*$ which satisfies \eqref{eqn:buyingfromoptisbetter}.
	But first, some notation.
	For each batch arrival order $\sigma(1), \ldots, \sigma(N)$, suppose that in round $i$ upon the arrival of $B_{\sigma(i)}$, $A$ buys the sets $\mathcal{C}_{i} \subseteq \mathcal{S}_{\Delta \times H}$.
	These $\mathcal{C}_1, \ldots, \mathcal{C}_i$ together cover each $B_{\sigma(i)}$ upon its arrival, and $c(A, \sigma) = \sum_i |\mathcal{C}_i|$.
	Let $\overline{\mathcal{C}_i} := \{S \in \bigcup_{i' \leq i} \mathcal{C}_{i'}: S \cap  B_{\sigma(i)} \neq \emptyset \}$ be the collection of sets which $A$ uses to cover $B_{\sigma(i)}$.
	We say a set $S_{lj} \in \mathcal{S}_{\Delta \times H}$ is \textit{live} in round $i$ if it intersects with all batches $B_{\sigma(1)}, \ldots, B_{\sigma(i)}$ which have arrived so far.
	All sets are live to start, and once a set is not-live it will never be live again; note that the liveness of $S_{lj}$ in round $i$ is a property of its label $l$ and the batch order $\sigma$. 
	In each round $i$ we will \textit{match} each set $S^* \in \mathcal{C}_i^*$ which $A^*$ buys with some set $S \in \bigcup_{i' \leq i} \mathcal{C}_{i'}$.
	We will maintain that at most one $S^*$ is matched to each $S$, and that a matched pair of sets is never unmatched.
	
	Let $A^*$ operate by running $A$ in the background. 
	For each round $i$ with incoming batch $B_{\sigma(i)}$, if $A^*$ already covers $B_{\sigma(i)}$ upon arrival then $A^*$ does nothing.
	Otherwise for each $j^* \in \opt(H)$ for which $A^*$ has not already bought a copy which is live in round $i$, $A^*$ identifies an unmatched set $S_{lj} \in \overline{\mathcal{C}}_i$ which $A$ is using to cover $B_{\sigma(i)}$, buys the set $S_{lj^*}$ (with the same label), and matches $S_{lj^*}$ with $S_{lj}$.
	
	It is immediate that $c(A^*, \sigma) \leq c(A, \sigma)$, since $A^*$ matches every set which it buys to a set which $A$ buys. 
	Therefore we need only show that $A$ is feasible in each round $i$; that is, it never runs out of unmatched sets.
	
	To see this, first note that projecting $\overline{\mathcal{C}}_i$ onto $H$ gives a feasible cover, and so $|\overline{\mathcal{C}}_i| \geq k$. 
	In particular, this means that $A^*$ succeeds in the first round $i=1$. 
	Next observe that in any round $i>1$, algorithm $A^*$ has bought at most one $S_{lj^*}$ which is live for each given $S_{j^*} \in \opt(H)$. 
	This is because it only buys $S_{lj^*}$ for $j^*$ for which it does not currently have a live copy, and no sets go from not-live to live.
	Also note that any set which shares a label with some $S_{lj} \in \overline{\mathcal{C}}_i$ is live in round $i$. 
	Therefore any $S_{lj^*}$ matched to $S_{lj} \in \overline{\mathcal{C}}_i$ at the beginning of round $i$ are live, since $A^*$ ensures that matched sets share labels.
	Let $\Gamma_i$ be the collection of these $S_{lj^*}$, and let $t:= |\Gamma_i|$.
	Since $A^*$ maintains at most one live set for each $S_{j^*} \in \opt(H)$ at once, the $j^*$ for $S_{lj^*} \in \Gamma_i$ are distinct. 
	Therefore to cover $B_{\sigma(i)}$, $A^*$ must buy sets $S_{lj^*}$ for the $k - t$ remaining $S_{j^*} \in \opt(H)$ not represented in $\Gamma_i$ with live labels $l$.
	Fortunately there are $|\overline{\mathcal{C}}_i| - t \geq k - t$ unmatched sets with live labels which $A$ has bought for $A^*$ to choose from, and so $A^*$ never gets stuck.
	
	Therefore $A^*$ is feasible for every round $i$, and so \eqref{eqn:buyingfromoptisbetter} holds.
	
	We now lower-bound the performance of $A^*$.
	Since $A^*$ is lazy and $\opt(H)$ is a minimal cover for $H$, for each arriving batch $B_{\sigma(i)}$ the algorithm $A^*$ buys exactly one $S_{lj^*}$ corresponding to each $S_{j^*} \in \opt(H)$ for which it does not already have a live copy.
	Therefore we can analyze the expected number of copies of $S_{j^*}$ which $A^*$ buys over the randomness of the batch arrival order for each $S_{j^*} \in \opt(H)$ independently.
	
	Fix some $S^* = S_{j^*} \in \opt(H)$, and let $C_N$ denote the expected number of copies of $S^*$ which $A^*$ buys, where the expectation is taken over the batch arrival order $\sigma \sim S_N$. 
	In the first round $i=1$, $A^*$ buys some $S_{lj^*}$ with label $l$, and uses this copy until the first batch arrives for which $l$ is no longer live; let $P(l)$ denote the number of batches $B_{\sigma(1)} \ldots, B_{\sigma(P(l))}$ for which $l$ remains live.
	Once $l$ is no longer live, $A^*$ must choose another copy $S_{l'j^*}$ for one of the remaining live $l'$. 
	This is a sub-instance of the problem it faced at $i=1$. 
	
	We will prove that $C_N \geq H_N / 2$ by induction on $N$ (where here $H_N$ denotes the $N^{th}$ harmonic number). 
	By definition we have that $C_1 = 1$, and so this holds for $N=1$. 
	Now assume the claim $C_n \geq H_n / 2$ holds for all integers $n \in [0, N-1]$.
	Using the observations above, we can express $C_N$ by the following recurrence:
	\begin{align*}
		C_N &= 1 + \sum_{j = 1}^N \prob{P(l) = j} \cdot C_{N - j} \\
		&= 1 +\frac{1}{N} \sum_{n = 0}^{N-1}  (H_N - H_{n}) \cdot C_{n},
		\intertext{where we take $H_0 = 0$. By computing $C_{N-1}$ and substituting, we then obtain}
		C_N &= \frac{1}{N} +\frac{N-1}{N} C_{N-1} + \frac{1}{N^2} \sum_{n = 0}^{N-1}  C_n \\
		&\geq \frac{1}{N} +\frac{N-1}{2N}\left(H_N - \frac{1}{N} \right) + \frac{1}{2N} (H_N - 1)\\
		&\geq \frac{1}{2} H_N,
	\end{align*}		
	where here we used that $\sum_{n=0}^{N-1} H_n = N (H_N - 1)$.
	Since $H_N \geq \log N$, we therefore have
	\[
	\expectover{\sigma \sim S_N}*{\left\vert\left\{S_{lj^*} \in \bigcup_i \mathcal{C}^*_i\right\}\right\rvert} = C_N \geq  \frac{1}{2} \log N.
	\]
	Since this analysis holds for each of the $k$ such sets $S_{j^*} \in \opt(H)$, the claim follows from linearity of expectation.
\end{proof}

\section{Pseudocode}

\label{sec:appendix2}

In this section we give pseudocode for secondary algorithms in this paper.

\subsection{The Exponential-Time Set Cover Algorithm}

First we give the algorithm from \cref{sec:exptime}.
\begin{algorithm}[H]
	\caption{\textsc{SimpleLearnOrCover}}
	\label{alg:expunit}
	\begin{algorithmic}[1]
		\State Initialize $\mathfrak{T}^0 \leftarrow \binom{\mathcal{S}}{k}$ and $\mathcal{C}^0 \leftarrow \emptyset$.
		\For{$\thist=1,2\ldots, n$}
		\State $v^{\thist} \leftarrow$ $\thist^{th}$ element in the random order.
		\If {$v^{\thist}$ uncovered}
		\State Choose $\mathcal{T} \sim \mathfrak{T}^{\lastt}$ uniformly at random, and choose $T \sim \mathcal{T}$ uniformly at random.
		\State Add $\mathcal{C}^{\thist} \leftarrow \mathcal{C}^\lastt \cup \{S, T\}$ for any choice of $S$ containing $v^{\thist}$.
		\EndIf
		\State Update $\mathfrak{T}^{\thist} \leftarrow \{ \mathcal{T} \in \mathfrak{T}^{\lastt} : v^{\thist} \in \bigcup \mathcal{T} \}$.
		\EndFor 
	\end{algorithmic}
\end{algorithm}

\subsection{The Set Cover Algorithm for Unit Costs}

Next, we give a slightly simplified version of algorithm from
\cref{sec:gen_cost} in the special case of unit costs. We give it here
to illustrate the essential simplicity of our algorithm in the
unit-weight case. (Much of the complication comes from managing the
non-uniform set costs.)

	\begin{algorithm}[H]
	\caption{\textsc{UnitCostLearnOrCover}}
	\label{alg:unitcost}
	\begin{algorithmic}[1]
		\State Initialize $x_S^{0} \leftarrow \frac{1}{m}$, and $\mathcal{C}^0 \leftarrow \emptyset$.
		\For{$\thist=1,2\ldots, n$}
		\State $v^{\thist} \leftarrow$ $\thist^{th}$ element in the
			random order.
		\If {$v^{\thist}$ not already covered}
		\State Sample one set $R^{\thist}$ from distribution $x$, update $\mathcal{C}^{\thist} \leftarrow \mathcal{C}^{\lastt} \cup \{R^{\thist}\}$.
		\If {$\sum_{S \ni v^{\thist}} x^{\lastt}_S < 1$}
		\State For every set $S \ni v^{\thist}$, update $x^{\thist}_S \leftarrow e \cdot x^{\lastt}_S$. 
		\State Normalize $x^{\thist} \leftarrow x^{\thist} / \|x^{\thist}\|_1$.
		\Else
		\State $x^{\thist} \leftarrow x^{\lastt}$.
		\EndIf
		\State Let $S_{v^{\thist}}$ be an arbitrary set containing $v^{\thist}$. Add $\mathcal{C}^{\thist} \leftarrow \mathcal{C}^{\thist} \cup \{S_{v^{\thist}}\}$.
		\EndIf
		\EndFor 
	\end{algorithmic}
\end{algorithm}

\section{Discussion of Claim in \cite{grandoni2008set}}

\label{sec:ggl}

\cite{grandoni2008set} write in passing that \setcov in ``the random permutation model (and hence any model where elements are drawn from an unknown distribution) [is] as hard as the worst case''. 

Our algorithm demonstrates that the instance of \cite{korman2004use} witnessing the $\Omega(\log^2 n)$ lower bound in adversarial order can be easily circumvented in random order. Korman's instance is conceptually similar to our hard instance for the batched case from \cref{sec:lb_ext} but with batches shown in order, deterministically.

A natural strategy to adapt the instance to the random order model is to duplicate the elements in the $i^{th}$ batch $C^{b-i}$ times for a constant $C>1$ (where $b$ is the number of batches). This ensures that elements of early batches arrive early with good probability. Indeed, this is the strategy used for the online Steiner Tree problem in random order (see e.g \cite{gupta2020random}). However, for the competitive ratio lower bound of \cite{korman2004use}, which is $\log b \log s$, to be $\Omega(\log^2 n)$, the number of batches $b$ and the number of distinct elements per batch $s$ must be polynomials in the number of elements $n$. In this case the total number of elements after duplication is $n' = O(C^{\poly(n)})$, which degrades the $\Omega(\log b \log s)$ bound to doubly logarithmic in $n'$.

\section{Connections to Other Algorithms}

\label{sec:disussion}

In this section we discuss connections between \nameref{alg:gencost} and other algorithms. Whether these perspectives are useful or merely spurious remains to be seen, but regardless, they provide interesting context.

\textbf{Projection interpretation of \cite{buchbinder2019k}.} In \cite{buchbinder2019k} it is shown that the original \osetcov algorithm of \cite{alon2003online,buchbinder2009online} is equivalent to the following. Maintain a fractional solution $x$. On the arrival of every constraint $\langle a_i, x \rangle \geq 1$, update $x$ to be the solution to the convex program
\begin{align}
	\begin{array}{lll}
		\argmin_{y} & \KL{y}{x} \label{eq:bn_proj}\\
		\text{s.t.} &A^{\leq i} y \geq 1  \\
		&y \geq 0,
	\end{array}
\end{align}
where $A^{\leq i}$ is the matrix of constraints up until time $i$. Finally, perform independent randomized rounding online. The KKT and complementary slackness conditions ensure that the successive fractional solutions obtained this way are monotonically increasing, and in fact match the exponential update rule of \cite{buchbinder2009online}.

Interestingly, we may view \nameref{alg:gencost} as working with the same convex program but \textit{with an additional cost constraint} (recall that $\beta$ is our guess for $\copt$).
\begin{align}
	\begin{array}{lll}
		\argmin_{y} & \KL{y}{x} \\
		\text{s.t.} &A^{\leq i} y \geq 1  \\
		&\langle c, y \rangle  \leq \beta \\
		&y \geq 0,
	\end{array}
\end{align}
The extra packing constraint already voids the monotonicity guarantee of \eqref{eq:bn_proj} which we argue is a barrier for \cite{buchbinder2009online} in \cref{sec:ro_lb,sec:bn_lb}. Furthermore, instead of computing the full projection, we fix the Lagrange multiplier of the most recent constraint $\langle a_i, x \rangle  \geq 1$ to be $\kappa_{v_i}$ (the cost of the cheapest set containing this last element), the multipliers of the other elements' constraints to $0$, and the multiplier of the cost packing constraint such that $y$ is normalized to cost precisely $\beta$. Hence we only make a partial step towards the projection, and sample from the new fractional solution regardless of whether it was fully feasible. 

\textbf{Stochastic Gradient Descent (SGD).} Perhaps one reason \roosc is easier than the adversarial order counterpart is that, in the following particular sense, random order grants access to a stochastic gradient. 

Consider the unit cost setting. Given fractional solution $x$, define the function 
\[f(x) := \sum_v \max\left(0, 1- \sum_{S \ni e} x_S\right),\]
in other words the fractional number of elements uncovered by $x$. Clearly we wish to minimize $f$, and if we assume that every update to $x$ coincides with buying a set (as is the case in our algorithm), we wish to do so in the smallest number of steps.

The gradient of $f$ evaluated at the coordinate $S$ is
\[[\nabla f]_S = - |S \cap U^{\thist}|.\]
On the other hand, conditioning on the next arriving element $v$ being uncovered, the random binary vector $\chi^v$ denoting the set membership of $v$ has the property
\[\expectover{v}*{\chi^v_S} = \frac{|S \cap U^{\thist}|}{|U^{\thist}|},\]
meaning $\chi^v$ is a scaled but otherwise unbiased estimate of $\nabla f$. Since \nameref{alg:gencost} performs updates to the fractional solution using $\chi_v$, it can be thought of as a form of stochastic gradient descent (more precisely of stochastic mirror descent with entropy mirror map since we use a multiplicative weights update scheme, see e.g. \cite{bubeck2014convex}).  

One crucial and interesting difference is that SGD computes a gradient estimate at every point to which it moves, whereas our algorithm is only allowed to query the gradient at the vertex of the hypercube corresponding to the sets bought so far. This analogy with SGD seems harder to argue in the non-unit cost setting, where the number of updates to the solution is no longer a measure of the competitive ratio.
	
\end{document}